\documentclass[11pt,b5paper]{article}
\usepackage[latin1]{inputenc}
\usepackage{amsmath,amsthm}%,amsfonts,amssymb}
\usepackage{graphicx}
\usepackage{hyperref}
\usepackage[round]{natbib}
\usepackage{setspace}
\usepackage[dvipsnames]{xcolor}
\usepackage{enumitem}
\usepackage{authblk}
\usepackage[margin=2.8cm]{geometry}
\usepackage{thmtools}
\usepackage{thm-restate}
\usepackage{cleveref}
\usepackage{mathtools}%,mathrsfs}
\usepackage[stable]{footmisc}
\usepackage{booktabs}

\declaretheorem[name=Lemma]{lemma}

\declaretheorem[name=Proposition]{prop}
\usepackage[font=footnotesize,labelfont=sc]{caption}

\usepackage[utopia]{mathdesign}
\usepackage{microtype}

\theoremstyle{remark}
\newtheorem{rem}{Remark}

\newcommand\DD{\mathscr{D}}
\newcommand\PP{\mathscr{P}}
\newcommand\WW{\mathscr{W}}

\newcommand\Ax{\mathscr{A}}
\newcommand\NNN{\mathbb{N}}
\newcommand\RRR{\mathbb{R}}

\newcommand\ZZZ{\mathbb{Z}}

\newcommand\inv{^{-1}}
\newcommand\Pair[1]{\langle#1\rangle}
\newcommand\marginal{^s}
\DeclareMathOperator\Tot{Tot}

\DeclareMathOperator\QAM{QAM}

\setlist[description]{labelindent=\parindent,leftmargin=\parindent}

\newcommand\CL[1][]{{\mathrm{CL}_{#1}}}
\newcommand\TU{{\mathrm{TU}}}
\newcommand\AU{{\mathrm{AU}}}
\newcommand\VV[1]{{\mathrm{VV#1}}}
\newcommand\PR[1][]{{\mathrm{PR}_{#1}}}
\newcommand\MDT[1][]{{\mathrm{MDT}}}
\newcommand\QAA[1][]{{\mathrm{QAA}}}
\newcommand\GRD[1][]{{\mathrm{GRD}}}
\newcommand\CLL[1][]{{\mathrm{CLL}_{#1}}}
\newcommand\BRD[1][]{{\mathrm{BRD}}}
\newcommand\sCL[1][]{_{\CL[#1]}}
\newcommand\sTU{_{\TU}}
\newcommand\sAU{_{\AU}}
\newcommand\sVV[1]{_{\VV{#1}}}
\newcommand\sPR[1][]{_{\PR[#1]}}
\newcommand\sMDT[1][]{_{\MDT}}
\newcommand\sQAA[1][]{_{\QAA}}
\newcommand\sGRD[1][]{_{\GRD}}
\newcommand\sBRD[1][]{_{\BRD}}
\newcommand\sCLL[1][]{_{\CLL[#1]}}

\newcommand{\Pop[1]}{{#1}}
\newcommand{\Size[1]}{\ensuremath{\lvert\Pop[#1]\rvert}}

\newcommand{\Avg[1]}{\overline{\Pop[#1]}}

\newcommand*\Value{V}

\newcommand*\SetOfPops{\PP}

\DeclareMathOperator\MAD{MD}

\newif\ifanon
\anonfalse

\title{Non-Additive Axiologies in Large Worlds}
\ifanon
\author{}
\date{}

\else
\author{Christian Tarsney\thanks{Global Priorities Institute, Faculty of Philosophy, University of Oxford. Comments welcome: christian.tarsney@philosophy.ox.ac.uk, teru.thomas@oxon.org.} \space and Teruji Thomas$^*$}
\date{Version 1.0, Sept 2020 \\ \vspace{1mm}
(\href{https://sites.google.com/view/christiantarsney/research}{Latest version here.})}
\fi

\begin{document}
\maketitle

\begin{abstract}
%Is the value of a population simply the sum of values contributed by each member of that population?
Is the overall value of a world just the sum of values contributed by each value-bearing entity in that world? \textit{Additively separable} axiologies (like total utilitarianism, prioritarianism, and critical le\-vel views) say `yes', but non-additive axio\-logies (like average utilitarianism, rank-discounted utilitarianism, and variable value views) say `no'. %The distinction between additive and non-additive axiologies 
This distinction is practically important: additive axiologies support `arguments from astronomical scale' which suggest (among other things) that it is overwhelmingly important for humanity to avoid premature extinction and ensure the existence of a large future population, %the far future is vastly more important than the near future, 
while non-additive axiologies need not. %apparently do not. 
We show, however, that when there is a large enough `background population' %of individuals 
unaffected by our choices, a wide range of non-additive axiologies converge in their 
%[TT Sep 29] removing
%practical 
implications with some additive axiology---for instance, average utilitarianism converges to critical-level utilitarianism and various egalitarian theories converge to prioritiarianism. We further argue that %the real-world background population is large and diverse enough to make these limit results practically significant
%we are in fact \textit{in} the relevant limit---what is true in the limit is true of us/in practice
real-world background populations may be large enough to make these limit results practically significant.
%are large enough to make these limit results practically significant, at least in some contexts.
%so that what is true in the limit is true of us. 
This means that arguments from astronomical scale, and other arguments in practical ethics that seem to presuppose additive separability, may be truth-preserving in practice whether or not we accept additive separability as a basic axiological principle.
%, among other things, that arguments from astronomical scale can succeed even if we reject additivity.
\end{abstract}

%todo REUSE MATERIAL FROM EARLIER ABSTRACT? An (additively) separable population axiology is any view on which the overall value of a population can be expressed as the sum of values assigned to discrete value-bearing entities that compose the population based on the intrinsic properties of those locations. The distinction between separable axiologies (like total utilitarianism, prioritarianism, and critical level views) and non-separable axiologies (like average utilitarianism, rank-discounted utilitarianism, and variable value views) is one of the most important dividing lines in population ethics. In this paper, however, we argue that this divide has less practical significance than one might think: For practical purposes, in the world we happen to live in, all plausible non-separable axiologies behave like separable axiologies. 2-3 SENTENCE EXPLANATION. These convergence results have important practical implications: In particular, `arguments from astronomical scale' supporting, e.g., the conclusion that we should attribute overwhelming moral significance to the very-long-run consequences of our actions, which appear to presuppose separability, are in fact likely to succeed on nearly any plausible axiology. 

%\doublespacing
%\onehalfspacing

\section{Introduction}
\label{section-Introduction}

%todo Opening sentence options:
%The world is large.
%The world is big.
%The world we live in is both large and populous.
%The world we inhabit is both large and populous.
%The world we find ourselves in is both large and populous.
%The world is large, along multiple dimensions.
%The world is large, in more ways than one.
%We live in a world that is both large and populous.
%We find ourselves in a world that is both large and populous.

The world we live in is both large and populous. Our planet, for instance, is 4.5 billion years old and has has borne life for roughly 4 billion of those years. At any time in its recent history, it has played host to billions of mammals, trillions of vertebrates, and quintillions of insects and other animals, along with countless other organisms.  Our galaxy contains hundreds of billions of stars, many or most of which have planets of their own. The observable universe is billions of light-years across, containing hundreds of billions of galaxies---and is probably just a small fraction of the universe as a whole. 
It may be, therefore, that our biosphere is just one of many (perhaps infinitely many). 
Finally, the \textit{future} is potentially vast: our descendants could survive for a very long time,
%millions of years (and perhaps far longer), 
and might someday settle a large part of the accessible universe, gaining access to a vast pool of resources that would enable the existence of astronomical numbers of beings with diverse lives and experiences. %future beings with lives and experiences.

These facts have ethical implications. Most straightforwardly, the potential future scale of our civilization %scale of future human-originating civilization 
suggests that %, to whatever extent we can, 
it is extremely important to shape the far future for the better. This view has come to be called \textit{longtermism}, and its recent proponents include \cite{bostrom2003astronomical,bostrom2013existential}, %(who focuses on the long-run value of reducing existential risks to human civilization) 
\cite{beckstead2013overwhelming,beckstead2019brief}, %(who gives a general defense of longtermism and explores a range of %potential 
%practical implications), 
\cite{cowen2018stubborn}, %(who focuses on the long-run value of economic growth)
\cite{greavesMScase}, and \cite{ord2020precipice}. There are many ways in which we might try to positively influence the far future---e.g., building better and more stable institutions, shaping cultural norms and moral values, or accelerating economic growth. But one %a 
particularly obvious concern %, on which we will have particular reason to focus, 
is ensuring the \textit{long-term survival} of our civilization, by avoiding civilization- or species-ending `existential catastrophes' from sources like nuclear weapons, climate change, biotechnology, and artificial intelligence.\footnote{The importance of avoiding existential catastrophe is especially emphasized by \cite{bostrom2003astronomical,bostrom2013existential} and \cite{ord2020precipice}.} Longtermism in general, and the emphasis on existential catastrophes in particular, have major revisionary practical implications if correct, e.g., suggesting the need for major reallocations of resources and collective attention \citep[pp.\@ 57ff]{ord2020precipice}.

All these recent defenses of longtermism appeal, in one way or another, to the \textit{astronomical scale} of the far future. For instance, Beckstead's central argument starts from the premises that `Humanity may survive for millions, billions, or trillions of years' and `If humanity may survive may survive for millions, billions, or trillions of years, then the expected value of the future is astronomically great' \citep[pp.\@ 1--2]{beckstead2013overwhelming}. %todo* Add more quotations in this paragraph, e.g. from Bostrom, Greaves & MacAskill?
Importantly for our purposes, the astronomical scale of the far future most plausibly results from the \textit{astronomical number of individuals} who might exist in the far future: while the far future population might consist, say, of just a single galaxy-spanning individual, the futures that typically strike longtermists as most worth pursuing %fighting for 
involve a very large number of individuals with lives worth living (and conversely, the futures most worth avoiding involve a very large number of individuals with lives worth not living).

Under this assumption, we can understand arguments like Beckstead's as instantiating the following schema.

\begin{description}
	\item[Arguments from Astronomical Scale]~\\
		Because far more welfare subjects or value-bearing entities are affected by A than by B, we can make a much greater %bigger 
		difference to the overall value of the world by focusing on A rather than B.
\end{description} 
Beckstead and other longtermists take this schema and substitute, for instance, `the long-run trajectory of human-originating civilization' for A and `the (non-trajectory-shaping) events of the next 100 years' for B. To illustrate the scales involved, \cite{bostrom2013existential} estimates that if we manage to settle the stars, our civilization could ultimately support at least $10^{32}$ century-long human lives, or $10^{52}$ subjectively similar lives in the form of simulations. Since only a tiny fraction of those lives would exist in the next century or millennium, it seems \textit{prima facie} plausible that even comparatively minuscule effects on the far future (e.g., small changes to the average welfare of the far-future population, %the welfare of a small fraction of the far-future population, 
or to its size, or to the probability that it %such an astronomical population 
comes to exist in the first place) would be vastly more important than any effects we can have on the more immediate future.\footnote{Thus, for instance, in reference to the $10^{52}$ estimate, Bostrom claims that `if we give this allegedly lower bound...a mere 1 per cent chance of being correct, we find that the expected value of reducing existential risk by a mere \textit{one billionth of one billionth of one percentage point} is worth a hundred billion times as much as a billion human lives' \citep[p.\@ 19]{bostrom2013existential}.}

Should we find arguments from astronomical scale persuasive? That is, does the fact that A affects vastly more individuals than B give us strong reason to believe, in general, that A is vastly more important than B? 
%Supposing that the far future might involve astronomical numbers of individuals with either good or bad lives, does this imply, as longtermists suggest, that shaping the far future is \textit{overwhelmingly important}, for an agent whose aim is to make the world a better place? %todo ...to the greatest extent possible? ...whose aim is to do the most good?
%[TT Sep 28] Adding some hedging/clarification here - the underlying problem is that the `all else equal' clause, which only enters in at the end of the paragraph, is far from clearly true (at least: but for the sheer numbers).
Although there are many possible complications, the sheer numbers make these arguments quite strong if
%The answer seems to be a fairly straightforward `yes', if 
we accept an axiology (a theory of the value of possible worlds or states of affairs) according to which the overall value of the world is simply a \textit{sum} of values contributed by each individual in that world---e.g., the sum of individual welfare levels. In this case, the effect that some intervention has on the overall value of the world scales linearly with the number of individuals affected (all else being equal), and so astronomical scale %todo ADD? "(in terms of the number of individuals affected)"
implies astronomical importance. %Since there are potentially vastly more individuals in the far future than in the near future, even relatively small changes to the far future (e.g., changing the welfare of a small fraction of those individuals by a small amount) can make an astronomical difference to the overall value of the world.

But can the overall value of the world be expressed as such a sum? This question represents a crucial dividing line in axiology, between axiologies that are \textit{additively separable} (hereafter usually abbreviated `additive') and those that are not. Additive axiologies allow the value of a world to be represented as a sum of values independently contributed by each value-bearing entity in that world, while non-additive axiologies do not. 
%Additive axiologies allow the value of a population to be represented as a sum of values associated with each welfare subject (or other value-bearing entity), while non-additive axiologies do not. 
For example, \textit{total utilitarianism} claims that the value of a world is simply the sum of the welfare of every welfare subject in that world, and is therefore additive. On the other hand, \textit{average utilitarianism}, which identifies the value of a world with the \textit{average} welfare of all welfare subjects, is non-additive.

%[TT Sep 28] Making myself happy
When we consider non-additive axiologies, the force of arguments from astronomical scale becomes much less clear, 
%Non-additive axiologies are often unimpressed by arguments from astronomical scale, 
especially in variable-popu\-lation contexts (i.e.~when comparing possible populations of different sizes). They therefore represent a challenge to the case for longtermism and, more particularly, to the case for the overwhelming importance of avoiding existential catastrophe. 
%[TT Sep 28] I found the following a bit confusing, so trying to clarify; previous version below
As a stylized illustration: suppose that there are $10^{10}$ existing people, all with welfare 1. We can either ($O_1$) leave things unchanged, ($O_2$) improve the welfare of all the existing people from 1 to 2, or ($O_3$) create 
%[TT Sep 28] making variable 
some number $n$ of new people with welfare 1.5. Total utilitarianism, of course, tells us to choose $O_3$, as long as $n$ is sufficiently large. But average utilitarianism---while agreeing that $O_3$ is better than $O_1$ and that the larger $n$ is, the better---nonetheless prefers $O_2$ to $O_3$ no matter how astronomically large $n$ may be. Now, additive axiologies can disagree with total utilitarianism here if they claim that adding people with welfare $1.5$ makes the world \emph{worse} instead of better; but the broader point is that they will almost always claim that the difference in value between $O_3$ and $O_1$ becomes astronomically large (whether positive or negative) as $n$ increases---bigger, for example, than the difference in value between $O_2$ and $O_1$. Non-additive axiologies, on the other hand, need not regard $O_3$ as making a big difference to the value of the world, regardless of $n$. Again, average utilitarianism  agrees with total utilitarianism that $O_3$ is an improvement over $O_1$, but regards it as a \textit{smaller} improvement than $O_2$, even when it affects vastly more individuals.

Thus, the abstract question of additive separability seems to play a
%[TT] `crucial' may be overstating it
crucial role with respect to arguably the most important practical question in population ethics: the relative importance of (i) ensuring the long-term survival of our civilization %, and therefore 
and its ability to support a very large number of future individuals with lives worth living vs.\@ (ii) improving the welfare of the present population.
%average welfare (or the distribution of welfare) within the present population.

The aim of this paper, however, is to show that under certain circumstances, a wide range of non-additive axiologies converge in their
%[TT Sep 28] Removing `practical', because of prospective oughts. NB grammar a bit funny
%practical 
implications with some counterpart additive axiology. This convergence has a number of interesting 
%[TT Sep 28] varying wording
consequences,
%implications, 
but perhaps the most important is that non-additive axiologies can inherit the scale-sensitivity of their additive counterparts. This makes arguments from astronomical scale less reliant on the controversial %axiological 
assumption of additive separability. It thereby increases the robustness of the practical case for the overwhelming importance of the far future and of avoiding existential catastrophe. %increasing the robustness of the case for longtermism and for the importance of existential risk.

%allows arguments from astronomical scale to go through without relying on the assumption of additive separability, and increasing the robustness of the case for longtermism.

%for conclusions like the overwhelming importance of the long-run future and existential risk mitigation go through without relying on the assumption of additive separability.

%[TT Sep 28] Sign-posting seemed off
Our starting place is the observation that, according to non-additive axiologies, 
%To preview our argument: Non-additive axiologies 
%%(but not additive axiologies)   
%have the feature that,
%[TT Sep 28] Making more axiological
which of two outcomes is better can depend on the welfare of the people unaffected by the choice between them. 
%when we can affect only part of the overall population, what we ought to do can depend on features of the sub-population that we \textit{can't} affect. 
%[TT Sep 28] Focusing on outcomes, trimming parenthetical, adding footnote
That is, suppose we are comparing two populations $\Pop[X]$ and $\Pop[Y]$.\footnote
{
    We follow the tradition in population ethics that `populations' are individuated not only by which people they contain, but also by what their welfare levels would be. (However, in the formalism introduced in section \ref{section-FormalSetup}, the populations we'll consider are \emph{anonymous}, i.e. the identities of the people are not specified.)
}
%todo #CUT the parenthetical?
%That is, suppose we are given a choice between two populations $\Pop[X]$ and $\Pop[Y]$.  ($\Pop[X]$ and $\Pop[Y]$ might contain all the same individuals, so that we are simply choosing between different welfare profiles. Or they might contain entirely different individuals, so that we are choosing which group to bring into existence. Or they might be partly overlapping.) %todo #CUT the parenthetical?
%[TT Sep 28] moving away from `choice' language
And suppose that, besides $X$ and $Y$, there is some `background population' $\Pop[Z]$ that would  exist either way.
%And suppose there is some `background population' $\Pop[Z]$ that will 
%be unaffected by our choice. 
($\Pop[Z]$ might include, for instance, past human or non-human welfare subjects on Earth, faraway aliens, or present/future welfare subjects who are simply unaffected by our present choice.) %like present and near-future wild animals
Non-additive axiologies allow that whether 
%[TT Sep 28] de-teleologizing: we're evaluating outcomes, not acts
$X$-and-$Z$ is better than $Y$-and-$Z$ can
%it is better to bring about $\Pop[X]$ or $\Pop[Y]$ can %sometimes 
depend on facts about $\Pop[Z]$.\footnote{The role of background populations in non-separable axiologies has received surprisingly little attention, but has not gone entirely unnoticed. In particular, \citeauthor{budolfsonMSwhy} (ms) consider the implications of background populations for issues related to the `Repugnant Conclusion' (see \S \ref{section-RepugnantAddition} below). And, as we discovered while revising this paper, an argument very much in the spirit of our own (though without our formal results) was elegantly sketched several years ago in a blog post by Carl Shulman \citep{shulman2014population}.}

%This paper explores a particularly significant aspect of this sensitivity to background populations. First, we show that ...

%todo #REINSERT? In this paper, first, we prove several results to the effect that non-additive axiologies of various kinds converge with counterpart additive axiologies in the large-background-population limit. That is, for any non-additive axiology $\Ax$ in one of these classes, for any choice between `foreground' populations $\Pop[X]$ and $\Pop[Y]$, and given a background population $\Pop[Z]$ with certain fixed characteristics (like average welfare or the relative frequencies of different welfare levels), we can identify an additive axiology $\Ax'$ such that $\Ax$ will agree with $\Ax'$ about whether to add $\Pop[X]$ and $\Pop[Y]$ to the population, as long as $\Pop[Z]$ is large enough.
%given a large enough background population, various classes of non-additive axiologies asymptotically agree 
%we show that as the size of the background population increases, many non-additive axiologies asymptotically agree with some additive counterpart axiology
%a wide range of non-additive axiologies converge in their practical implications with some counterpart additive axiology

%In this paper, first, we show that ...
%[TT Sep 28]
With this in mind, our
%Our 
argument has two steps. First, we prove several results to the effect that, %show that 
in the large-background-population limit (i.e., as the size of the background population $\Pop[Z]$ tends to infinity), non-additive axiologies of various types 
%[TT Sep 28] fixing grammar, and trimming
converge with counterpart additive axiologies. Thus, these axiologies are
%each converge with some counterpart additive  axiology in their ranking of potential additions to the population. %(i.e., of the possible additional populations $\Pop[X]$ and $\Pop[Y]$) Thus, these axiologies become
effectively additive in the presence of sufficiently large background populations. 
%asymptotically agree with some counterpart additive axiology, so that these axiologies become effectively additive in the presence of large enough background populations. 
Second, we argue that the background populations in real-world choice situations are, at a minimum, substantially larger than the present and near-future human population. This provides some \textit{prima facie} reason to believe that non-additive axiologies of the types we survey will agree closely with their additive counterparts in practice. More specifically, we argue that real-world background populations are large enough to substantially increase the importance that average utilitarianism (and, more tentatively, variable value views) assign to avoiding existential catastrophe.
%large enough to make these limit results practically significant, so that at least for many practical purposes, 
%we are \textit{actually in} this limiting case, so that for at least very many real-world choices,
%these non-additive axiologies are effectively equivalent to their additive counterparts. And among these practical purposes, we will argue, is the question of how to prioritize between the long-term future %/survival of civilization 
%and the interests of the current generation. 
Thus, our arguments suggest, it is not merely the potential scale of the \textit{future} that has important ethical implications, but also the scale of the world as a whole---in particular, the scale of the background population. %of welfare subjects unaffected by our choices.
%[TT] Nice ^

%[TT Sep 28] Can't stomach starting sentences with \S; easiest tweak is spelling out "section"
The paper proceeds as follows: section  \ref{section-FormalSetup} introduces some formal concepts and notation, while section \ref{section-Additivity} formally defines additive separability and describes some important classes %prominent examples 
of additive axiologies. In sections  \ref{section-AveragistViews}--\ref{section-EgalitarianViews}, we survey several important classes of non-additive axiologies %todo #REINSERT? (averagist and asymptotically averagist views in \S \ref{section-AveragistViews}, and egalitarian views in \S \ref{section-EgalitarianViews})
and show that they become additive %converge with an additive counterpart axiology
in the large-background-population limit.
%[TT] TODO I really like this wording^ - much clearer than convergence mumbo jumbo
In section \ref{section-RealBackgroundPopulation}, we consider the size and other characteristics of real-world background populations and, in particular, argue that they are at least substantially larger than the present human population.
%argue that the background populations in real-world choice situations are very large %todo* (and diverse)
%at least relative to the present human population, %"present terrestrial population"?
%which provides \textit{prima facie} evidence that our limit results are significant in practice.
%is large and diverse enough that we are in the relevant limit case---i.e., what is true in the limit is true of us. 
In sections \ref{section-CausalDomainRestriction}--\ref{section-CountingSomeLess}, we answer two objections: that we should simply ignore background populations for decision-making purposes, and that we should apply `axiological weights' to non-human welfare subjects that reduce their contribution to the size of the background population. %we consider the objection that we should simply ignore the background population for purposes of moral decision-making, and offer several replies
Section \ref{section-ExistentialCatastrophe} %describes the practical implications of large-background-population limiting behavior for 
considers how real-world background populations affect the importance of avoiding existential catastrophe %"tradeoffs between present welfare and future population size"?
according to average utilitarianism and variable-value views. Section \ref{section-PracticalImplcations} briefly describes three more potential implications of our results: %...of the preceding arguments
they make it harder to avoid (a generalization of) the Repugnant Conclusion, help us to extend non-additive axiologies to infinite-population contexts, and suggest that agents who accept non-additive axiologies may be vulnerable to a novel form of manipulation. Section \ref{section-Conclusion} is the conclusion.

\section{Formal setup}
\label{section-FormalSetup}

All of the population axiologies we will consider evaluate worlds based only on the number of welfare subjects at each welfare level.
We will consider only worlds containing a finite \emph{total} number of welfare subjects (except in \S \ref{section-InfiniteEthics}, where we consider the significance of our results for infinite ethics). We will also set aside worlds that contain \emph{no} welfare subjects, simply because some theories of population axiology, like average utilitarianism, do not evaluate such empty worlds.  

Thus for our purposes a \emph{population} is  a non-zero, finitely supported function from the set $\WW$ of all possible welfare levels to the set $\ZZZ_+$ of all non-negative integers, specifying the number of welfare subjects at each level. Despite this formalism, we'll say that a welfare level $w$ \emph{occurs} in a population $X$ to mean that $X(w)\neq 0$. 
An \emph{axiology} $\Ax$ is a strict partial order $\succ_{\Ax}$ on the set $\PP$ of all  populations, with `$X\succ_\Ax Y$' meaning that population $X$ is better than population $Y$ according to $\Ax$. Almost all the axiologies we will consider in this paper can be represented by a \textit{value function} $V_\Ax \colon \SetOfPops \to \RRR$, meaning that $X\succ_\Ax Y \iff V_\Ax(X)>V_\Ax(Y)$.

To illustrate this formalism, %todo #CUT?
the \emph{size} $\Size[X]$ of a population $X$ is simply the total number of welfare subjects:
\[
	\Size[X]\coloneqq \sum_{w\in\WW} X(w).
\]
Similarly, the total welfare is
\[
	\Tot(X)\coloneqq \sum_{w\in\WW} X(w)w.
\]
Of course, the definition of $\Tot(X)$ only makes sense on the assumption that we can add together welfare levels, and in this connection we generally assume that $\WW$ is given to us as a set of real numbers. 
With that in mind, the average welfare 
\[
	\Avg[X]\coloneqq \Tot(X)/\Size[X]
\]
is also well-defined.  

\section{Additivity}
\label{section-Additivity}

We can now give a precise definition of additive separability.

If $\Pop[X]$ and $\Pop[Y]$ are populations, then let $\Pop[X + Y]$ be the population obtained %simply 
by adding together the number of welfare subjects at each welfare level in $\Pop[X]$ and $\Pop[Y]$. That is, for all $w \in \WW$, $(\Pop[X + Y])(w) = \Pop[X](w) + \Pop[Y](w)$. An axiology is \emph{separable} if, for any populations $X$, $Y$, and $Z$, 
\[
	X+Z\succ Y+Z \iff X\succ Y.
\]
 This means that in comparing $X+Z$ and $Y+Z$, one can ignore the %common `background population' 
 shared sub-population $Z$. Separability is entailed by the following more concrete condition:
\begin{description}
	\item[Additivity]~\\
		An axiology $\Ax$ is \emph{additively separable} (or \emph{additive} for short) iff it can be represented by a value function of the form 
		\[	V_\Ax(X)=\sum_{w\in\WW} X(w)f(w)
		\]
		with $f\colon\WW\to\RRR$. Thus the value of $X$ is given by transforming the welfare of each welfare subject by the function $f$ and then adding up the results.
\end{description} 

%[TT] TODO The following paragraph should be revisited later - seems to both over- and under-emphasize separability.
In the following discussion, we will sometimes want to focus on the distinction between additive and non-additive axiologies, and sometimes on the distinction between separable and non-separable axiologies. While an axiology can be separable but non-additive, none of the views we will consider below have this feature. So for our purposes, the additive/non-additive and separable/non-separable distinctions are more or less extensionally equivalent.\footnote{For a detailed discussion of separability principles in population ethics, see \citeauthor{thomasFCseparability} (forthcoming).}

We will consider three categories of additive axiologies in this paper, which we now introduce in order of increasing generality.
First, there is \textit{total utilitarianism}, which identifies the value of a population with its total welfare.\footnote{Total utilitarianism is arguably endorsed (with varying degrees of clarity and explicitness) by classical utilitarians like \cite{hutcheson1725inquiry}, \cite{bentham1789principles}, \cite{mill1863utilitarianism}, and \cite{sidgwick1907methods}, and has more recently been defended by \cite{hudson1987diminishing}, \cite{singer2014point}, and \citeauthor{gustafssonFCpopulation} (forthcoming), among others.} %todo* More cites - Tannsjo? Ng? Huemer? Norcross?

\begin{description}
	\item[Total Utilitarianism $(\TU)$] 
		\[
			V\sTU(\Pop[X]) = \Tot(X) = \sum_{w\in\WW} X(w)w=\Avg[X]\Size[X].
		\]
%[TT] CT often uses only the last notation; I'm not sure what's best, but from an algebraic point of view the total is more natural than the average...
\end{description}

An arguable drawback of $\TU$ is that it implies the so-called `Repugnant Conclusion' \citep{parfit1984reasons}, that for any two positive welfare levels $w_1 < w_2$, for any population in which everyone has welfare $w_2$, there is a better population in which everyone has welfare $w_1$. The desire to avoid the Repugnant Conclusion is one motivation for the next class of additive axiologies, \textit{critical-level} theories.\footnote{Critical-level views have been defended by \cite{blackorby1997critical,blackorby2005population}, among others.} %todo* More cites?

\begin{description}
	\item[{Critical-Level Utilitarianism $(\CL)$}] 
		%[TT] CT likes writing \Size[X]\Avg[X] and the like for total welfare, but I tend to find this distracting, so I'm trying out \Tot(X) in various places.
		\[
			V\sCL(\Pop[X]) = \sum_{w\in\WW} X(w) (w - c) = \Tot(X) - c\Size[X] = (\Avg[X]-c)\Size[X]
		\] 
		for some constant $c \in \WW$ (representing the `critical level' of welfare above which adding an individual to the population constitutes an improvement), generally but not necessarily taken to be positive.
\end{description}
We sometimes write `$\CL[c]$' rather than merely `$\CL$' to emphasize the dependence on the critical level. $\TU$ is a special case of $\CL$,  namely, the case with critical level $c=0$. Note that, as long as $c$ is positive, $\CL$ avoids the Repugnant Conclusion since adding lives with very low positive welfare makes things worse rather than better.\footnote{But a positive critical level also brings its own, arguably greater drawbacks---e.g., the Strong Sadistic Conclusion \citep{arrhenius2000impossibility}.}

Another arguable drawback of both $\TU$ and $\CL$ is that they give no priority to the less well off---that is, they assign the same marginal value to a given improvement in someone's welfare, regardless of how well off they were to begin with. We might intuit, however, that a one-unit improvement in the welfare of a very badly off individual has greater moral value than the same welfare improvement for someone who is already very well off. This intuition is captured by \textit{prioritarian} theories.\footnote{Versions of prioritarianism have been defended by \cite{weirich1983utility}, \cite{parfit1997equality}, \cite{arneson2000luck}, and \cite{adler2009future,adler2011well}, among others. \textit{Sufficientarianism}, which by our definition will count as a special case of prioritarianism, has been defended by \cite{frankfurt1987equality} and \cite{crisp2003equality}, among others.}

\begin{description}
	\item[Prioritarianism $(\PR)$] 
		\[
			V\sPR(\Pop[X]) = \sum_{w\in\WW} X(w) f(w)
		\]
		for some function $f\colon  \WW \to\RRR$ (the `priority weighting' function) that is 
%[TT] TODO - optimize weak vs strict issues with t:GenEgal
		concave and strictly 
		increasing.  
\end{description}
$\CL$ is a special case of $\PR$ where $f$ is linear, and $\TU$ is a special case where $f$ is linear and passes through the origin. Note also that our definition of the prioritarian family of axiologies is very close to our definition of additive separability, just adding the conditions that $f$ is strictly increasing and weakly concave.

\section{Averagist and asymptotically averagist views}
\label{section-AveragistViews}

In this section and the next, we consider two categories of non-additive axiologies and show that, in the presence of large enough background populations, they converge with some additive axiology. In this section, we show that average utilitarianism and related views converge with $\CL$, where the critical level is the average welfare of the background population. In the next section, we show that various non-additive egalitarian views converge with PR.

First, though, what do we mean by converging to an additive (or any other) axiology? The claim makes sense relative to a specified \emph{type} of background population, e.g., all those having a certain average level of welfare.

\begin{description}
	\item[Convergence]~\\
		 Axiology $\Ax$ converges to $\Ax'$ relative to background populations of type $T$, if and only if, for any populations $X$ and $Y$, if $Z$ is a sufficiently large population of type $T$, then 
		\[
			X+Z\succ_{\Ax'} Y+Z \implies X+Z\succ_\Ax Y+Z. %[CT Sept 30] I still think \implies is waaaayyyy too long, but I'll bow to your superior notational wisdom on this!
		\]
\end{description} 
Of course, if $\Ax'$ is separable, the last implication can be replaced by 
\[
	X\succ_{\Ax'} Y \implies X+Z\succ_{\Ax} Y+Z. 
\]
We can, in other words, compare $X+Z$ and $Y+Z$ with respect to $\Ax$ by comparing $X$ and $Y$ with respect to $\Ax'$---if we know that $Z$ is a sufficiently large population of the right type. 

Note two ways in which this notion of convergence is fairly weak. First, what it means for $Z$ to be `sufficiently large' can depend on $X$ and $Y$. Second, the displayed implication need not be a biconditional; thus, when $\Ax'$ does not have a strict preference between $X+Z$ and $Y+Z$ (e.g., when it is indifferent between them), convergence to $\Ax'$ does not imply anything about how $\Ax$ ranks of those two populations. %specifically, it leaves open the possibility that $X+Z$ is better than $Y+Z$ according to $\Ax$, but neither is better than the other according to $\Ax'$. 
Because of this, every axiology converges to the trivial axiology according to which no population is better than any other. Of course, such a result is uninformative, and we are only interested in convergence to more discriminating axiologies. Specifically, we will only ever consider axiologies that satisfy the Pareto principle (which we discuss in \S\ref{section-twoFactorEgalitarianism}). %todo #CUT this paragraph, or the last three sentences?
%[TT] Seems kind of informative to me - would keep for now. Added cross-ref

%todo* #REINSERT? And third, convergence is easier to obtain the more narrowly we specify the type $T$ of background populations under consideration.

\subsection{Average utilitarianism}

Average utilitarianism, as the name suggests, identifies the value of a population with the average welfare level of that population.\footnote{Average utilitarianism is often discussed but rarely endorsed. It has its defenders, however, including \cite{hardin1968tragedy}, \cite{harsanyi1977morality}, and \cite{pressman2015defence}. %todo* More cites? Economists?? 
\cite{mill1863utilitarianism} can also be read as an average utilitarian (see fn.\@ 2 in \citeauthor{gustafssonFCintuitive} (forthcoming)), though the textual evidence for this reading is not entirely conclusive.

As with all evaluative or normative theories---but perhaps more so than most---average utilitarianism confronts a number of choice points that generate a minor combinatorial explosion of possible variants. \cite{hurka1982average,hurka1982more} identifies three such choice points which generate at least twelve different versions of averagism. The view we have labeled $\AU$ (which Hurka calls A1) strikes us as the most plausible, but our main line of argument could be applied to many other versions. %todo* [Do some thinking to confirm that this is right...] 
Versions of averagism that only care about the \textit{future} population do present us with a challenge, which we discuss in \S \ref{section-CausalDomainRestriction}.}

\begin{description}
	\item[Average Utilitarianism $(\AU)$]
		\[
			\Value\sAU(\Pop[X]) = \Avg[X] = \sum_{w\in\WW} \frac{X(w)}{\Size[X]}w.
		\]
\end{description}

%First result: AU asymptotically agrees with CL given a large enough background population---specifically, with the version of CL where the critical level $c$ is the average welfare of the background population. EXPLAIN WHAT `ASYMPTOTICALLY AGREE WITH' MEANS.

%Here is our first result.
\noindent
Our first result describes the behavior of AU as the size of the background population tends to infinity. 

\begin{restatable}[]{thm}{AveragismThm}\label{t:AU}
	Average utilitarianism converges to CL$_c$, relative to background populations with average welfare $c$. In fact, for any populations $\Pop[X],\Pop[Y],\Pop[Z]$, 
	if $\Avg[Z]=c$ and 
	\begin{equation}\label{eq:avg1}
		\Size[Z] > \frac{\Size[X] V\sCL[c](\Pop[Y]) - \Size[Y]V\sCL[c](\Pop[X])}{V\sCL[c](\Pop[X]) - V\sCL[c](\Pop[Y])} 
	\end{equation} 	
	then 
	$V\sCL[c](\Pop[X]) > V\sCL[c](\Pop[Y]) \implies V\sAU(X + Z) > V\sAU(Y + Z)$. 
\end{restatable}

%[CT] We currently have numbers ((1), (2), etc at the right margin) on just a couple equations in the paper. Not sure if there was a reason for numbering these in particular, but if we don't need the numbering for anything, probably best to forego it entirely?

Proofs of all theorems are given in the appendix.
%
%[TT] The following discussion is currently superfluous
%todo #REINSERT? Note that, since $\CL[c]$ is additive, $V\sCL[c](\Pop[X]) > V\sCL[c](\Pop[Y]) \Leftrightarrow V\sCL[c](\Pop[X + Z]) > V\sCL[c](\Pop[Y + Z])$, so the conclusion can be equivalently stated as $V\sCL[c](\Pop[X + Z]) > V\sCL[c](\Pop[Y + Z]) \Rightarrow \Avg[X + Z] > \Avg[Y + Z]$. Note also that the conclusion, being conditional rather than biconditional, doesn't cover the case where $V\sCL[c](\Pop[X]) = V\sCL[c](\Pop[Y])$.
%
Discussion of this and other results is deferred to \S \ref{section-PracticalImplcations}.

\subsection{`Variable value' views}

Some philosophers have sought an intermediate position between total and average utilitarianism, acknowledging that increasing the size of a population (without changing average welfare) can count as an improvement, but holding that additional lives have \textit{diminishing marginal value}. The most widely discussed version of this approach is the 
%[TT] Not really clear to me what the extension of `variable value view should be, but this seems to fit with the literature
\textit{variable value} view.\footnote
{
	These views were introduced by \cite{hurka1983value}. Variable Value I is also discussed by \cite{ng1989future} under the name `Theory X$'$'. 
%[Also say something about `geometrism' \cite{sider1991might}?]
%[TT] Probably not?
}
	It is useful to distinguish two types of this view, the second more general than the first. 
\begin{description}
	\item[Variable Value I $(\VV1)$]~\\
		$V\sVV1(\Pop[X]) = \Avg[X] g(\Size[X])$, where $g \colon \ZZZ_+ \to \RRR_+$ is increasing, concave, 
%[TT] Have to rule out g=0 somehow, so adding (alternatively could do strictly increasing or strictly concave, but certainly the latter doesn't seem fundamental
%[TT] TOCHECK double check for consistency whether we use of "strict" for increasing, concave
        non-zero,
		and bounded above.
\end{description}

\begin{description}
	\item[Variable Value II $(\VV2)$]~\\ 
	$V\sVV2(\Pop[X]) = f(\Avg[X])g(\Size[X])$, where $f \colon \mathbb{R} \to \mathbb{R}$ is differentiable and strictly increasing, and $g\colon \ZZZ_+ \to \RRR_+$ is increasing, concave, 
%[TT] Have to rule out g=0 somehow, so adding (see above)
    non-zero,
	and bounded above. 
\end{description}

Sloganistically, variable value views can be `totalist for small populations' (where $g$ may be nearly linear), but must become `averagist for large populations' (as $g$ approaches its upper bound). It is therefore not entirely surprising that, in the large-background-population limit, VV1 and VV2 display the same behavior as AU, converging to a critical-level view with the critical level given by the average welfare of the background population.

\begin{restatable}[]{thm}{VVThm}\label{t:VV}
	Variable value views converge to $\CL[c]$ relative to background populations with average welfare $c$. 
\end{restatable}

%todo ADD SOME KIND OF PROOF SKETCH? MAYBE SOMETHING LIKE THIS EARLIER THING? "Argument: Since $g$ is bounded above, VV1 converges to (`asymptotically agrees with'?) Averagism as $|\Pop[Z]| \to \infty$. And Averagism converges to a Critical Level view as $|\Pop[Z]| \to \infty$. And it seems like the relevant convergence relation (`agrees with, in any particular case, for all $n$ greater than some $n$') is transitive."

For the broad class of variable value views, we cannot give the sort of threshold for $\Size[Z]$ that we gave for $\AU$, above which the ranking of $\Pop[X] + \Pop[Z]$ and $\Pop[Y] + \Pop[Z]$ must agree with the ranking given by $\CL[{\Avg[Z]}]$. For instance, because $g$ can be \textit{any} function that is strictly increasing, strictly concave, and bounded above, variable value views can remain in arbitrarily close agreement with totalism for arbitrarily large populations, so if $\TU$ prefers one population to another, there will always be \textit{some} variable value theory that agrees. In the case of $\VV1$,  we can say that if \textit{both} $\TU$ and $\AU$ prefer $\Pop[X]$ to $\Pop[Y]$, then all $\VV1$ views will as well (see Proposition \ref{p:VV1mix} in the appendix), and so whenever $\TU$ and $\CL[{\Avg[Z]}]$ have the same strict preference between $\Pop[X]$ and $\Pop[Y]$, the threshold given in Theorem \ref{t:AU} holds for $\VV1$ as well. 
For $\VV2$, we cannot even say this much.\footnote{
What we can say about $\VV2$ is the following: when $\Avg[X] > \Avg[Y]$, $\Size[X] \geq \Size[Y]$, and $f(\Avg[X]) \geq 0$, VV2 is guaranteed to prefer $\Pop[X]$ to $\Pop[Y]$. Similarly, when $\Avg[X] > \Avg[Y]$, $\Size[Y] \geq \Size[X]$, and $f(\Avg[Y]) \leq 0$, VV2 is guaranteed to prefer $\Pop[X]$ to $\Pop[Y]$. (These claims depend only on the fact that $f$ is strictly increasing and $g$ is increasing.) 
%[TT] i.e. weakly, in the second case
So in any case where the population preferred by $\CL[{\Avg[Z]}]$ is larger and has average welfare to which $\VV2$ assigns a non-negative value, or the population dispreferred by $\CL[{\Avg[Z]}]$ is larger and has average welfare to which $\VV2$ assigns a non-positive value, $\VV2$ will agree with $\CL[{\Avg[Z]}]$ whenever $\AU$ does.
}

%todo ADD SOME KIND OF PROOF SKETCH? E.G. "As a sketch: Since $g$ is strictly increasing and strictly concave, it becomes asymptotically linear as $|\Pop[Z]| \to \infty$ (since $\Pop[X]$ and $\Pop[Y]$ have smaller and smaller effects on $\overline{\Pop[X] + \Pop[Z]}$ and $\overline{\Pop[Y] + \Pop[Z]}$) respectively. So as $|\Pop[Z]| \to \infty$, VV2 converges with VV1 which converges with Averagism which converges with the Critical Level view where $c = \overline{\Pop[Z]}$."

\section{Non-additive egalitarian views}
\label{section-EgalitarianViews}

%A second motivation for non-additivity is \textit{egalitarianism}. 
A second category of non-additive axiologies are motivated by egalitarian considerations. Whether adding some individual to a population, or increasing the welfare of an existing individual, will increase or decrease equality depends on the welfare of other individuals in the population, so it is easy to see why concern with equality %egalitarianism %egalitarian concerns 
might motivate separability violations.

Egalitarian views have been widely discussed in the context of distributive justice for fixed populations, %todo* Add cites? 
but relatively little has been said about egalitarianism in a variable-population context. We are therefore somewhat in the dark as to which egalitarian views are most plausible in that context. But we will consider a few possibilities that seem especially promising, trying to consider each fork of two major choice points for variable-population egalitarianism.

The most important choice point is between (i) `two-factor'/`pluralistic' egalitarian views, which treat the value of a population as the sum of two (or more) terms, one of which is a measure of inequality, and (ii) `rank-discounting' views, which give less weight to the welfare of individuals who are better off relative to the rest of the population. These two categories of views are extensionally equivalent in the fixed-population context, but %can 
come apart in the variable-population context (\citeauthor{kowalczykMSequality}, ms).

\subsection{Two-factor egalitarianism}
\label{section-twoFactorEgalitarianism}

Among two-factor egalitarian theories, there is another important choice point between `totalist' and `averagist' views.

\begin{description}
	%[TT] Consider rewriting as (\Avg[X]-\alpha I(X))\Size[X] -- raises fewer questions!
	\item[Totalist Two-Factor Egalitarianism]~\\
		$V(\Pop[X]) = \Tot(X) - I(\Pop[X])\Size[X]$, where $I$ is some measure of inequality in $\Pop[X]$.
	\item[Averagist Two-Factor Egalitarianism]~\\
		$V(\Pop[X]) = \Avg[X] - I(\Pop[X])$, where $I$ is some measure of inequality in $\Pop[X]$.\footnote{One could also imagine variable-value two-factor theories (and two-factor theories that incorporate critical levels, priority weighting, etc., into their value functions), but we will set these possibilities aside for simplicity.}
\end{description}

Here, in each case, the second term of the value function can be thought of as a penalty representing the badness of inequality. Such a penalty could have any number of forms, but for the purposes of illustration we stipulate that $I(X)$ depends only on the \emph{distribution} of $X$, where this can be understood formally as the function $X/\Size[X]\colon\WW\to\RRR$ giving the proportion of the population in $X$ having welfare $w$. 
The \emph{degree} of inequality is indeed plausibly a matter of the distribution in this sense, and the \emph{badness} of inequality is then plausibly a function of the degree of inequality and the size of the population. The more substantial assumption is that the badness of inequality either scales linearly with the size of the population (for the totalist version of the view) or does not depend on population size (for the averagist version). 
%[TT] TODO I have an email from Kacper somewhere with references about the linearity - not sure it is useful, and not also worrying about it right now, but reminding myself there might be something to cite.

Now, we want to know what these theories do as $\Size[Z] \to \infty$. In the last section, we had to hold one feature of $\Pop[Z]$ constant as $\Size[Z]\to \infty$, namely, $\Avg[Z]$. Egalitarian theories, however, are potentially sensitive to the whole distribution of welfare levels in the population, and so to obtain limit results it is useful to hold fixed the whole distribution of welfare in the background population, i.e.~$D\coloneqq Z/|Z|$.
%fixed. Here, more precisely, the distribution is the \emph{proportion} of the population that has each welfare level. More generally, for any population $X$, and any set $W\subset\WW$ of welfare levels, write $D_X(W)$ for the proportion of welfare subjects in $X$ who have welfare levels in $W$:
%\[
%	D_X(W) = \sum_{w\in W} \frac{X(w)}{\Size[X]}.
%\]
%The function $D_X$ determines, and is determined by, the distribution of $X$.
We'll state the general result, and then give some examples.

%[TT] TODO See if there's a better way, e.g. I could just be missing something easy.
%[TT] Sep 27 - making this a bit messier. 
%PROBLEMS: 
%(i) Just saying that \Ax converges to an additive axiology is vacuous, since the trivial axiology is additive. 
%(ii) The easy abstract argument re Pareto and Pigou-Dalton only gives us weak-increasing and weakly-concave  - again compatible with the trivial axiology! - and more general than our current defo of "prioritarianism"
%SO:
%* stating the form of the weighting function
%* It's a non-trivial claim that \Ax converges to the thing with this weighting function
%* Just saying this f is weakly increasing and weakly concave 
%QUESTIONS:
%* Should we define prioritarianism to cover this case?
%* Is there an argument for "Strictly increasing" - etc?
%*(I think there is an argument for "strictly increasing sometimes" if e.g. we say that making the worst off people marginally better off decreases inequality.)

\begin{restatable}[]{thm}{GenEgalThm}\label{t:GenEgal}
	Suppose $V$ is a value function of the form
	$V(X)=\Tot(X) - I(X)\Size[X]$, or else $V(X)=\Avg[X] - I(X)$, where $I$ is a differentiable function of the distribution of $X$. Then the axiology $\Ax$ represented by $V$ converges to an additive axiology relative to background populations with any given distribution $D$, with weighting function\footnote{Here $1_w\in\PP$ is the population with a single welfare subject at level $w$, and we use the fact that value functions of the assumed form can be evaluated directly on any finitely supported, non-zero function $\WW\to\RRR_+$, such as, in particular, $D$ and $D+t1_w$.}
\[f(w)=\lim_{t\to 0^+} \frac{V(D+t1_w)-V(D)}{t}. 
\]
If the Pareto principle holds with respect to $\Ax$, then $f$ is weakly increasing, and if Pigou-Dalton transfers are weak improvements, then $f$ is weakly concave. 
%Moreover, if the Pareto principle holds, and Pigou-Dalton transfers are weak improvements with respect to $\Ax$, then $\Ax$ converges to a form of prioritarianism.
\end{restatable} 

A few points in the theorem require further explanation. We will explain the relevant notion of \emph{differentiability} when it comes to the proof (see Remark \ref{rem:derivative} in the appendix); as usual, functions that are easy to write down tend to be differentiable, but it isn't automatic. The \emph{Pareto principle} holds that increasing anyone's welfare increases the value of the population. This principle clearly holds for prioritarian views (because the priority-weighting $f$ is assumed to be increasing), but it need not in principle hold for egalitarian views: conceptually, increasing someone's wellbeing might contribute so much to inequality as to be on net a bad thing. Still, the Pareto principle is generally held to be a desideratum for egalitarian views. Finally, a \emph{Pigou-Dalton transfer} is a total-preserving transfer of welfare from a better-off person to a worse-off person that keeps the first person better-off than the second. The condition that Pigou-Dalton transfers are at least weak improvements (they do not make things worse) is often understood as a minimal requirement for egalitarianism.  
%[TT] I might have made mistakes about weak vs strict betterness etc. [should now be consistent at least]
%[CT] Isn't the condition that Pigou-Dalton transfers are *at least weak* improvements satisfied by nearly every axiology (e.g., TU, CL, AU, VV, person-affecting views, etc)? Is it better to say that the characteristic feature of egalitarianism is that P-D transfers are at least weak *and sometimes strict* improvements, or something like that?
%[TT] I guess it depends whether you want utilitarianism to be a special case, etc. I've tweaked the wording, removing reference to prioritarianism, and saying "minimal requirement". Considered adding "and sometimes strict" but seemed unnecessary(?). Note we define prioritarianism to include utilitarianism. Hopefully less confusing now.
%[TT] I think it's all correct now, but see the comments on the theorem :(

To illustrate this result, let's consider two more specific families of egalitarian axiologies that instantiate the schemata of totalist and averagist two-factor egalitarianism respectively. 

For the first, we'll use a measure of inequality based on the \emph{mean absolute difference} ($\MAD$) of welfare, defined for any population $X$ as follows:
\begin{align*}
	\MAD(X) &:= \sum_{v,w\in\WW} \frac{X(w)X(v)}{\Size[X]^2}\left|w-v\right|.%\\
	% \StdDev(X) &= \sqrt{ \sum_{w\in\WW} \frac{X(w)}{\Size[X]}(w-\Avg[X])^2}
	%.
\end{align*}
$\MAD(X)$ represents the average welfare inequality between any two individuals in $\Pop[X]$. $\MAD(X)|X|$ can therefore be understood as measuring total pairwise inequality in $\Pop[X]$.
Consider, then,  the following totalist two-factor view:

%, the key question, of course, is what measure of inequality to use for $I$.  A measure of inequality is just a \textit{measure of dispersion}, applied to welfare levels.  There are many possible measures of dispersion, but perhaps the two most natural for a pluralist egalitarian to adopt are \textit{mean absolute difference} ($\MAD$) and \textit{standard deviation} ($\StdDev$).%

%This suggests four possible versions of two-factor egalitarianism for us to consider.

\begin{description}
	\item[Mean Absolute Difference Total Egalitarianism $(\MDT)$] 
		\[
			V\sMDT(\Pop[X]) = \Tot(X) - \alpha\MAD(\Pop[X])|\Pop[X]|
		\]
		where $\alpha \in (0,1/2)$ is a constant that determines the relative importance of inequality.\footnote{For $\alpha\geq 1/2$, equality would be so important that the Pareto principle would fail, i.e., it would no longer be true in general that increasing someone's welfare level increases the value of the population.} 
\end{description}

%[TT] Think we should cut these examples
%\begin{description}
%	\item[Mean Absolute Difference Average Egalitarianism (MDA)] $V_{MDA}(\Pop[X]) = \Avg[X] - \alpha\MAD(\Pop[X])$, where $\alpha \in (0,0.5)$ is a constant that determines the relative importance of inequality.
%\end{description}
%
%%[TT] Propose cutting SDT, because it violates Pareto
%\begin{description}
%	\item[Standard Deviation Total Egalitarianism (SDT)] 
%		\[
%			V_{SDT}(\Pop[X]) = \overline{\Pop[X]}|\Pop[X]| - 
%			\alpha\StdDev(\Pop[X])|\Pop[X]|,
%		\] 
%		where $\alpha \in (0,0.5]$ is a constant determining the relative importance of inequality. 
%%todo Define standard deviation in the description?
%\end{description}
%
%\begin{description}
%	\item[Standard Deviation Average Egalitarianism (SDA)] 
%		\[ V_{SDA}(\Pop[X]) = \Avg[X] - \alpha\StdDev(\Pop[X]),\]
%		where $\alpha \in (0,0.5]$ is a constant determining the relative 
%		importance of inequality. 
%%, where $\var(\Pop[X])$ gives the welfare variance of $\Pop[X]$, i.e., $\frac{\sum_{X} w_{\mathbf{x}}(x_i)^2}{\Size[X]} - \Avg[X]^2$
%\end{description}

Second, consider the following averagist two-factor view, which identifies overall value with a quasi-arithmetic mean of welfare:\footnote{See \citet{fleurbaey2010assessing} and \citet[Theorem 1]{mccarthy2015distributive} for axiomatizations of this type of egalitarianism, at least in fixed-population cases where the totalist/averagist distinction is irrelevant.}

%[TT] Propose adding (*if* we want another example); for now just sticking in the AU example (do we really need to list all the combinations separately?)
\begin{description}
	\item[Quasi-Arithmetic Average Egalitarianism $(\QAA)$] 
		\[
			V\sQAA(\Pop[X]) = \QAM(X)=g\inv\biggl({\sum_{w\in \WW} \frac{X(w)}{\Size[X]}g(w)}\biggr).
		\] 
		for some strictly increasing, concave function $g\colon\WW\to\RRR$.
\end{description}
Implicitly, the %relevant 
measure of inequality in $\QAA$ is $I(X)=\Avg[X]-\QAM(X)$, which one can show is a positive function, weakly decreasing under Pigou-Dalton transfers. In the limiting case where $g$ is linear, $\QAM(X)=\Avg[X]$. More generally, $\QAA$ is ordinally equivalent to an averagist version of prioritarianism. %todo* [Possibly insert veil of ignorance discussion currently in next subsection.]
%[TT] I added refs to Fleurbaey and McCarthy, above, which are probably better

\begin{restatable}[]{thm}{MDTThm}\label{t:MDT}
$\MDT$ converges to $\PR$, relative to background populations with a given distribution $D$. Specifically, $\MDT_\alpha$ converges to $\PR_f$, the prioritarian axiology whose %priority 
weighting function is
	\[
		f(w) = w - 2\alpha \MAD(w,D) + \alpha \MAD(D).
	\]
	Here $\MAD(w,D)\coloneqq \sum_{x\in\WW} D(x)|x-w|$ is the average distance between $w$ and the welfare levels occurring in $D$. 
\end{restatable}

\begin{restatable}[]{thm}{QAAThm}\label{t:QAA}
	$\QAA$ converges to $\PR$, relative to background populations with a given distribution $D$. Specifically, $\QAA_g$ converges to $\PR_f$, the prioritarian axiology whose %priority 
	weighting function is 
	\[
		f(w) = g(w)-g(\QAM(D)). 
	\]
\end{restatable}

\subsection{Rank discounting}

Another family of population axiologies that is often taken to reflect egalitarian motivations is \textit{rank-discounted utilitarianism} (RDU). The essential idea of rank-discounting is to give different weights %give more or less weight 
to marginal changes in the welfare 
%to the marginal welfare 
of different individuals, not based on their absolute welfare level (as prioritarianism does), but rather based on their welfare \textit{rank} within the population.

One potential motivation for RDU over two-factor views is that, because we are simply applying different positive weights to the marginal welfare of each individual, we clearly avoid any charge of `leveling down': unlike on two-factor views, there is nothing even \textit{pro tanto} good about reducing the welfare of a better-off individual---it is simply \textit{less bad} than reducing the welfare of a worse-off individual.\footnote{It is important to remember, however, that two-factor views with an appropriately chosen $I$, like those we considered in the last section, can avoid \textit{all-things-considered} leveling down: that is, while they may suggest that there is \textit{something good} about making the best off worse off, they never claim that it would be an all-things-considered improvement.}
%[TT] TODO ^ "at least when it comes to the value of outcomes" or whatever our way of saying this is

Versions of rank-discounted utilitarianism have been discussed and advocated under various names in both philosophy and economics, e.g.\@ by \cite{asheim2014escaping} and \cite{buchak2017taking}. %todo* Add more cites?
In these contexts, the RDU value function is generally taken to have the following form:
\begin{equation}\label{eq:RD}
	V(X)=\sum_{k=1}^{\Size[X]} f(k)X_k
\end{equation} 
where $X_k$ denotes the welfare of the $k$th worst off welfare subject in $X$, and $f\colon\NNN\to\RRR$ is a positive but decreasing function.\footnote
{\label{fn:gvsf}%
	Using the standard notation in this paper, one can alternatively write
	\[
		V(X)=\sum_{w\in\WW} \biggl(g\bigl(\sum_{v\leq w} X(v))-g(\sum_{v<w} X(v)\bigr)\biggr) w
	\]
	for some increasing, concave function $g\colon\RRR\to\RRR$ with $g(0)=0$.  The two presentations are equivalent if $g(k)=\sum_{i=1}^kf(k)$ or conversely $f(k)=g(k)-g(k-1)$.
}

However, these discussions often assume a context of fixed population size, and there are different ways one might extend the formula when the size is not fixed.
We will consider the most obvious approach, simply taking  equation \eqref{eq:RD} as a definition regardless of the size of $X$.\footnote
%[TT] Adding much shorter discussion of Buchak as footnote
{   An alternative approach would be to extend to variable-populations the `veil of ignorance' description of rank-discounting  described by Buchak (see also \citet[Example 2.9] {mccarthy2020utilitarianism}). However, on the most obvious way of doing this, the resulting view is coextensive with a two-factor egalitarian view and so falls under the purview of Theorem \ref{t:GenEgal} (even if it is conceptually different in important ways). 
}
 A view of this type, explicitly designed for a variable-population context, is set out in \cite{asheim2014escaping}. Simplifying slightly to set aside features irrelevant for our purposes, their view is as follows:
%[TT] Query: what's the simplification we're making? [CT] They allow a non-zero critical level and an increasing transformation of welfare (to "social utility" or whatever). (See p. 632.) [TT] TODO - think about in what sense this is irrelevant [CT] I want to say: "Adding in those complications, their view will still display the same sort of large-background-population limit behavior and is still subject to the same 'snobbishness' objection". But I haven't thought that through carefully.

\begin{description}
	\item[Geometric Rank-Discounted Utilitarianism $(\GRD)$] 
		\[
			V\sGRD(X) =  \sum_{k = 1}^{\Size[X]} \beta^k X_k
		\]
		for some $\beta \in (0,1)$.
\end{description}

Here, the rank-weighting function is $f(k) = \beta^k$. In general, since $f$ is assumed to be non-increasing and positive, $f(k)$ must asymptotically approach some limit $L$ as $k$ increases. For $\GRD$, $L=0$. But a simpler situation arises when $L>0$ (so that $f$ is bounded away from zero), and this is the case we will consider first, before returning to $\GRD$.
%[TT] BRD seems like one of the first things to cut

\begin{description}
	\item[Bounded Rank-Discounted Utilitarianism $(\BRD)$] 
		\[
			V\sBRD(X) =  \sum_{k = 1}^{\Size[X]} f(k) X_k
		\]
		for some 
%[TT] TODO make consistent - "non-increasing" = "decreasing" in our terminology?	
		non-increasing, positive function $f\colon\RRR\to\RRR$ that is eventually convex\footnote
		{
			That is, there is some $k$ such that $f$ is convex on the interval $(k,\infty)$. The assumption of eventual convexity is simply a technical assumption to be used in Theorem \ref{t:RD} below.
		}
		with asymptote $L>0$. 
\end{description}

In stating the result, we will need to restrict the foreground populations under consideration. %Suppose in general that $S$ is a set of populations.
\begin{description}
	\item[Convergence on \textbf{\textit{S}}]~\\
		 %TT's original wording: Axiology $\Ax$ converges to $\Ax'$ on $S$, relative to background populations of type $T$, if and only if, for any populations $X$ and $Y$ in $S$, if $Z$ is a sufficiently large population of type $T$, then 
		Axiology $\Ax$ converges to $\Ax'$, relative to background populations of type $T$, on a set of populations $S$, if and only if, for any populations $X$ and $Y$ in $S$, if $Z$ is a sufficiently large population of type $T$, then 
		\[
			X+Z\succ_{\Ax'} Y+Z \implies X+Z\succ_\Ax Y+Z.
		\]
\end{description} 
Having fixed a background distribution $D=Z/|Z|$, say that a population $X$ is \emph{moderate} with respect to $D$ if the the lowest welfare level in $X$ is no lower than the the lowest welfare level in $D$. In other words, for any $x\in\WW$ with $X(x)\neq 0$, there is some $z\in \WW$ with $z\leq x$ and $D(z)\neq 0$.  Then we can state the following result:

\begin{restatable}[]{thm}{positiveRDThm}\label{t:RD}
	$\BRD$ converges to $\TU$ relative to background populations with a given distribution $D$, on the set of populations that are moderate with respect to $D$.
	%TT's original statement of the theorem: Bounded rank-discounted utilitarianism converges to $\TU$ on the set of moderate populations with respect to a given background distribution $D$.
\end{restatable} 

When, as in $\GRD$, the asymptote of the weighting function $f$ is at $L=0$, the situation is subtler and appears to depend on the exact rate at which $f$ decays. We will consider only $\GRD$,  as it is the best-motivated example in the literature. 

In fact, GRD does \emph{not} converge to an additive, Paretian axiology on any interesting range of populations.  
Roughly speaking, this is because, as the background population gets larger, the 
weight given to the best-off individual in $X$ becomes arbitrarily small relative to the weight given to the worst-off---smaller than the relative weight given to it by any particular additive, Paretian axiology.
Nonetheless, it turns out that GRD \emph{does} converge to a \emph{separable}, Paretian axiology. We'll explain this carefully, but perhaps the most important take-away of this discussion will be that, given a large background population, GRD leads to some very strange and counterintuitive results. The limiting axiology will be \emph{critical level leximin}, defined by the following conditions: 
\begin{description}
	\item[{Critical Level Leximin $(\CLL[c])$}] ~
\begin{enumerate}[nosep]	
\item If $X$ and $Y$ have the same size, then $X\succ Y$ if and only if $X\neq Y$ and the least $k$ such that $X_k\neq Y_k$ is such that $X_k\succ Y_k$. 
\item If $X$ and $Y$ differ only in that $Y$ has additional individuals at welfare level $c$, then $X$ and $Y$ are equally good.\footnote{To compare $X$ and $Y$ in general, use the second condition to find populations $X'$ and $Y'$ that are equally as good as $X$ and $Y$ respectively, but such that $\Size[X']=\Size[Y']$, and then compare them using the first condition.
%[TT] I don't think we really need this either.
%Alternatively, the view can be formulated in the following way.  
%\label{fn:CLLdef}
%	 Let $\V$ be the space of finitely supported functions $\WW\to\RRR$. Introduce an ordering $\succ$ on $\V$: $f\succ g$ if and only if $f\neq g$ and $f(x)<g(x)$ if $x$ is the lowest welfare level at which $f(x)\neq g(x)$. Let $1_w\in\V$ be the function that takes value $1$ at $w$ and $0$ everywhere else. Given a population $X$, let $V(X)\in\V$ be the function $V(X)=\sum_{w\in\WW} X(w)(1_w-1_c)$. As one can see from this, critical level leximin is separable but not additive (insofar as additivity requires a real-valued, not function-valued, value function).
}
\end{enumerate}
\end{description}
In a sense, $\CLL[c]$ is simply a limiting case of prioritarianism, where the priority given to the less-well-off is infinite. In particular, although it is not additively separable in the narrow sense defined in \S \ref{section-FormalSetup}, which requires an assignment of real numbers to each individual, one can check that it is separable, and indeed one can show that it is additively separable in a more general sense, if we allow the contributory value of an individual's welfare to be represented by a vector rather than a single real number.\footnote{See \citet[Example 2.7]{mccarthy2020utilitarianism} for details in the constant-population-size case.}
%[TT] Could add details, see comments above; adding ref to utilitarianism paper - helpful but not 100% 

To state the theorem, fix a set $W\subset \WW$ of welfare levels. 
Say that a population $X$ is \emph{supported} on $W$ if $X(w)=0$ for all $w\notin W$.
And say that $W$ is \emph{covered} by a distribution $D=Z/\Size[Z]$ if and only if there is a welfare level in $Z$ between any two elements of $W$, a welfare level in $Z$ below every element of $W$, and welfare level in $Z$ above every element of $W$.   
%[TT] TODO - check/explain expression "welfare level in Z" 

\begin{restatable}[]{thm}{GDThm}\label{t:GD}
    Let $W\subset\WW$ be any set of welfare levels, and $D$ a population that covers $W$. 
	$\GRD$ converges to $\CLL[c]$ relative to background populations with distribution $D$, on the set of populations that are supported on $W$; the critical level $c$ is the highest welfare level occurring in $D$.
\end{restatable} 
%[TT] By assuming that $Z$ has someone above every level in $W$, I've made it so that the critical level can just be the max level in Z, even if that max level is below zero. 
%[TT] (I'm not sure we should get hung up on maximum generality!)

%[TT] Slightly hand-wavy discussion (previous version in comments below); using "often" as weasel word.
Critical level leximin has a number of extreme and implausible features; as the theorem suggests, these will often be displayed by $\GRD$ when there is a large background population.
For example, tiny benefits to worse-off individuals will often be preferred over astronomical benefits to even slightly better-off individuals; moreover, adding an individual to the population with anything less than the maximum welfare level in the background population will often make things worse overall.\footnote
%[TT] Not sure we really need the footnote, if we have the theorem. I'll try changing the footnote to cover both claims
{A toy example illustrates these phenomena, which are somewhat more general than the theorem entails. Suppose the background population consists of $N$ people at level $100$. Let $X$ consist of two people at level $99$; let $Y$ consist of one person at level $98$ and one at level $1000$; and let $Z$ consists of two people at level $99$ and one at $99.9$. We have $V\sGRD(X)-V\sGRD(Y)=\beta-\beta^2-900\beta^{N+2}$, which is positive if $N$ is large enough, in which case $X\succ\sGRD Y$, illustrating the first claim. On the other hand, $V\sGRD(X)-V\sGRD(Z)=0.1\beta^3-\beta^{N+3}$, again positive for $N$ large enough; then $X\succ\sGRD Z$, illustrating the second claim.}
In fact, according to $\CLL[c]$, it makes things worse to add one person slightly below the critical level along with any number of people above the critical level; because of this, $\GRD$ implies what we might call the `Snobbish Conclusion':
%[TT] Reformulating a bit to make it (maybe?) easier to understand
\begin{description}
	\item[Snobbish Conclusion]~\\ %todo "Obviously False Conclusion"?
		Suppose $X$ consists of one person with an arbitrarily good life, at level $w$, and any number of people with even better lives. Then there is some possible background population $Z$, in which the average welfare is far worse than $w$, and in which the very best lives are only slightly better than $w$, such that $Z+X$ is worse than $Z$.
\end{description}	
%[TT] I've trimmed the following discussion and toned it down a little; I'm referring to anti-natalism rather than neutrality, since it seems more directly relevant.
%[TT] I'm not sure this detailed level of discussion is really necessary - it seems like flogging a dead horse.
This seems crazy to us. We could just about understand the Snobbish Conclusion in the context of an anti-natalist view, according to which adding lives \emph{invariably} has negative value; but, according to $\GRD$, there are many possible background populations $Z$ such that $Z+X$ would be better than $Z$. We could also understand the view that adding good lives can make things worse if it lowers average welfare or increases inequality (e.g.\@ as measured by mean absolute difference or standard deviation).
But, again, that's not what's going on here. Instead, $\GRD$ implies that adding excellent lives makes things worse if the number of even slightly better lives already in existence happens to be sufficiently great, regardless of the other facts about the distribution.  In the limiting case, it makes things so much worse that it cannot be compensated by adding any number of even better lives.

\section{Real-world background populations}
\label{section-RealBackgroundPopulation}
%Titles
%Characterizing the real-world background population
%The real-world background population
%Real-world background populations

In the rest of the paper, we investigate the implications of the preceding results, and especially their practical implications for morally significant real-world choices. %The main \textit{potential} implication of these results is that, for practical purposes, various non-additive axiologies will reliably agree with their additive counterparts, and so certain general features of additive axiologies---like linear sensitivity to scale---can be regarded, in practice, as features of these non-additive axiologies as well.
%todo #REINSERT the sentence above?
As we have seen, how closely a given non-additive axiology agrees with its additive counterpart in some real-world choice situation depends on the size of the population that can be treated as `background' in %for purposes of 
that choice situation. And \textit{what} that additive counterpart will be (i.e., which version of $\CL$ or $\PR$) depends on the average welfare of the background population, and perhaps on its entire welfare distribution. In this section, therefore, we consider the size and (to a lesser extent) the welfare of real-world background populations.

%[TT Sep 29] Adding, because something like this seems obviously needed, somewhere , but not certain I'm saying the right thing (and I know it makes the intro of the next section slightly repetitious, but I feel we need to pull back the reader at this point). 
We note that nothing in this section (or the next two) shows conclusively that the background population is large {enough} for our limit results to be effective, but we do establish a \emph{prima facie} case for their relevance. In \S\ref{section-ExistentialCatastrophe}, we will seek firmer conclusions in a stylized case.
%[TT Sep 29] But is this quite the right description?

%[TT] TODO Hard for me to believe this detailed discussion is necessary!
We have so far taken the separation between `background' and `foreground' populations as given, but it will now be helpful to make these notions more precise. Given a choice between populations $\Pop[X_1], \Pop[X_2], ... \Pop[X_n]$, the population $\Pop[Z]$ that can be treated as background with respect to that choice is defined by $\Pop[Z](w) = \min_i \Pop[X_i](w)$. That is, the background population %in a given choice situation 
consists of the minimum number of welfare subjects at each welfare level who are guaranteed to exist regardless of the agent's choice. For this $\Pop[Z]$ and for each $\Pop[X_i]$, there is then a population $\Pop[X_i^*]$ such that $\Pop[X_i] = \Pop[X_i^*] + Z$. The choice between $\Pop[X_1], \Pop[X_2], ... \Pop[X_n]$ can therefore be understood as a choice between the foreground populations $\Pop[X_1^*], \Pop[X_2^*], \ldots, \Pop[X_n^*]$, in the presence of background population $\Pop[Z]$. %[CT] Should we say this earlier??

Clearly, this means that different real-world choices will involve different background populations. In particular, more consequential choices (that have far-reaching effects on the overall population) allow less of the population to be treated as background, whereas choices whose effects are tightly localized (or otherwise limited) %(e.g., only affecting the existence/welfare of a few individuals) 
may allow nearly the entire population to be treated as background. %Note also that \textit{later} choices will tend to have larger background populations and smaller foreground populations than \textit{earlier} choices?
But we can also define a `shared' background population for some \textit{set} of choice situations, by %simply 
considering %the set of 
all the overall populations that might be brought about by any \textit{profile} of choices in those situations. Thus we can speak, for instance, of the population that is `background' with respect to all the choices faced by present-day human agents, consisting of the minimum number of individuals at each welfare level that the overall population will contain whatever we all collectively do (perhaps simply equal to the number of individuals at each welfare level outside Earth's present future light cone).\footnote{Here and below, we assume a causal decision theory, which guarantees that causally inaccessible populations can be treated as `background'. %This assumption is purely for the sake of simplicity. 
How we can identify background populations, and how their practical significance changes, in the context of non-causal decision theories are interesting questions for future research.} %todo #CUT this footnote? Cut the whole preceding paragraph?
%which might 'only' contain all the welfare subjects outside of present-day Earth's future light cone).
%, but also of the background population with respect to a particular present-day choice, which may be much larger.

\subsection{Population size}
%Size of the background population

Past welfare subjects on Earth constitute the most obvious component of real-world background populations. Estimates of the number of human beings who have ever lived are on the order of $10^{11}$ \citep{kaneda2011how}, of whom only $\sim 7 \times 10^{9}$ are alive today. But of course \textit{Homo sapiens} are not the only welfare subjects. At any given time in the recent past, for instance, there are also many billions of mammals, birds, and fish being raised by humans for meat and other agricultural products. And given their very high birth/death rates, past members of these populations greatly outnumber present members.

But since human agriculture is a relatively recent phenomenon, farmed animals make only a relatively small contribution to the total background population. Wild animals make a far greater contribution. There are today, conservatively, $10^{11}$ mammals living in the wild, along with similar or greater numbers of birds, reptiles, and amphibians, and a significantly larger number of fish---conservatively $10^{13}$, and possibly far more.\footnote{For useful surveys of evidence on present animal population sizes, see \cite{tomasik2019how} and \cite{bar2018biomass} (especially pp.\@ 61-4 and Table S1 in the supplementary appendix).} %todo* Add more cites! [Cite estimates of mammal population.] [Cite estimates of fish population.]
This is despite the significant decline in wild animal populations in recent centuries and millennia as a result of human encroachment.\footnote{For instance, \citeauthor{smil2013harvesting} (\citeyear{smil2013harvesting}, p.\@ 228) estimates that wild mammalian biomass has declined by 50\% in the period 1900--2000 alone.} Inferring the total number of past mammals, vertebrates, etc from the number alive at a given time requires us to make assumptions about population birth/death rates. Unfortunately, we have not been able to find data that allow us to estimate overall birth/death rates for the wild mammal or wild vertebrate populations as a whole with any confidence. So we will simply adopt what strikes us as a very safely conservative assumption of 0.1 births/deaths per individual per year in wild animal populations (roughly corresponding to an average individual lifespan of 10 years). %, which is enough for our purposes
The actual rates are almost certainly much higher (especially for vertebrates), implying larger total past populations.

Being extremely conservative, then, we might suppose that all and only mammals are welfare subjects and that $10^{11}$ mammals have been alive on Earth at any given time since the K-Pg boundary event (the extinction event that killed the dinosaurs, $\sim 66$ million years ago), with a population birth/death rate of 0.1 per individual per year. This gives us a background population of $\sim 6.6 \times 10^{17}$ individuals. Being a bit less conservative (though perhaps still objectionably conservative), we might suppose that all and only vertebrates are welfare subjects and that $10^{13}$ vertebrates have been alive on Earth at any time in the last 500 million years (since shortly after the Cambrian explosion), with the same population birth/death rate of 0.1 per individual per year. This gives us a background population of $\sim 5 \times 10^{20}$ individuals.\footnote{In the name of conservatism, we are setting aside various hypotheses that might generate much larger background populations. First, of course, even the restriction to vertebrates excludes potential welfare subjects like crustaceans and insects. %todo* SAY MORE HERE? (1) NUMBER OF CRUSTACEANS AND INSECTS? (2) CITE DEBATES RE THEIR SENTIENCE?
Second, we Earthlings may not be the only welfare subjects. The observable %portion of the 
universe contains roughly 2 trillion galaxies \citep{conselice2016evolution}, %todo* ...and $10^{20}$ stars [cite]
and %while we do not yet know the size of the universe as a whole, it is likely to be many times larger than the observable universe \citep{vardanyan2011applications}. 
the universe as a whole is likely to be many times larger \citep{vardanyan2011applications}. The universe could therefore contain many other biospheres like Earth's. It might also contain advanced, spacefaring civilizations, which could support enormous populations on the order of $10^{30}$ individuals or more \citep{bostrom2003astronomical,bostrom2011infinite}. 
%todo #REINSERT? (\cite{bostrom2011infinite} estimates that if our own civilization settles the accessible universe quickly enough, we could (over the lifetime our own civilization) support at least $10^{34}$ biological human life-years, or $10^{54}$ human-life-year equivalents if we choose to invest in artificial rather than biological minds. The equivalent figures for other civilizations could be greater or less depending on how early they are able to begin settling the stars. It's not obvious, however, how to convert life-years into population size. For instance, in the `biological' scenario, advances in life extension might mean that $10^{34}$ human life-years correspond to far fewer than $10^{32}$ individual humans.) 
So the extraterrestrial background population could be many---indeed, indefinitely many---orders of magnitude larger than the populations of past mammals or vertebrates on Earth.}

\subsection{Welfare}
%Welfare
%Welfare distribution
%Average welfare and welfare distribution
%Average welfare in the background population

Anything we say about the distribution of welfare levels in the background population will of course be enormously speculative. So although the question has important implications, we will limit ourselves to a few brief remarks.

With respect to average welfare in the background population, two hypo\-theses seem particularly plausible. %...particularly natural %...seem most worth considering.

\begin{description}
	\item[Hypothesis 1] The background population consists mainly of small animals (whether terrestrial or extraterrestrial). Most of these animals have short natural lifespans, %(and most of them probably die in infancy), 
	so the average welfare level of the background population is very close to zero. If the capacity for positive/negative welfare scales with brain size (or related features like cortical neuron count), this would reinforce the same conclusion. It seems likely that average welfare in these populations will be negative, at least on a hedonic view of welfare \citep{ng1995towards,horta2010debunking}. %todo Other citations here? \citep[Ch.\@ 4]{dawkins1995river}?
	These assumptions together would imply, for instance, that $\AU$, VV1 and VV2 converge to a version of $\CL$ with a slightly negative critical level (perhaps very similar in practice to $\TU$). %Depending on how close the critical level is to zero, this might have very similar practical implications with convergence to $\TU$.
	
	\item[Hypothesis 2] The background population mainly consists of the members of advanced alien civilizations. If, for instance, the average biosphere produces $10^{23}$ wild animals over its lifetime, but one in a million biospheres gives rise to an interstellar civilization that produces $10^{35}$ individuals on average over \textit{its} lifetime, then the denizens of these interstellar civilizations would greatly outnumber wild animals in the universe as a whole. Under this hypothesis, given the limits of our present knowledge, all bets are off: average welfare of the background population could be very high \cite[pp.\@ 235--9]{ord2020precipice}, %todo Cite other stuff here? E.g. Bostrom, "Letter from Utopia"?
	very low \citep{sotala2017superintelligence}, %todo Other potential cites? \citep{tomasik2019risks,althaus2019reducing} (though not clear that they're actually describing scenarios where *average* welfare would be low). Google Scholar <"astronomical suffering"> brings up some other potentially citable stuff.
	or anything in between. %todo* [cite orthogonality thesis, `Long-term trajectories of human civilization', people worried about Malthusian scenarios---e.g., Hanson?]
	
\end{description}

%\subsection{Welfare distribution}
%Welfare distribution
%Distributional characteristics of the background population

With respect to the distribution of welfare more generally, we have even less to say. There is clearly a non-trivial degree of welfare inequality in the background population---compare, for instance, the lives of a well-cared-for pet dog and a factory-farmed layer hen. Self-reported welfare levels in the contemporary human population indicate substantial inequality (see for instance \cite{helliwell2019world}, Ch.\@ 2), and while \textit{contemporary} humans need not belong to the background population with respect to present-day choice situations, it seems safe to infer that there has been substantial welfare inequality in human populations in at least the recent past. 
%We have far more data on present-day humans than any other population, and these data suggest substantial levels of welfare inequality
%a fairly broad range in both life satisfaction and balance of positive/negative affect
%, insofar as we are able to measure these things 
For non-human animals, of course, we do not even have self-reports to rely on, and so any claims about the distribution of welfare are still more tentative. But there is, for instance, some literature on farm animal welfare that suggests significant inter-species welfare inequalities (e.g.\@ \citeauthor{norwood2011compassion} (\citeyear{norwood2011compassion}, pp.\@ 224--9), \cite{browning2020talk}). %todo* Confirm Browning citation!

That said, it could still turn out that the background population is dominated by welfare subjects who lead fairly uniform lives---e.g., by small animals who almost always experience lifetime welfare slightly below 0, or by members of alien civilizations that converge reliably on some set of values, social organization, etc., that produce enormous numbers of individuals with near-equal welfare. %But of course this is all pure speculation in which we should put very little stock.

\section{Objection 1: Causal domain restriction}
\label{section-CausalDomainRestriction}

We have shown that various non-additive axiologies converge to additive axiologies in the large-background-population limit.  But proponents of non-additive views might wish to avoid drawing practical conclusions from these results. After all, much of the point of being, say, an average 
utilitarian rather than a critical-level utilitarian is to reach the right 
practical conclusions in cases where $\AU$ seems more plausible than $\CL$.  %average utilitarianism strikes one as more plausible than critical-level utilitarianism
That point is defeated  if, in practice, $\AU$ is nearly  indistinguishable from $\CL$.

The simplest way to avoid the implications of our limit results is to claim that, for decision-making purposes, agents should simply ignore most or all of the background population. %, applying the correct axiology only to the populations affected by her choice. %(i.e., those individuals whose existence and/or welfare depends on her her choice).
%So the partisan of a  non-additive axiology might respond: `Why not just ignore the background  population?  After all, if your choices have no effect on their welfare,  aren't they just practically irrelevant---shouldn't we apply our axiology  only to the potential populations that the agent \textit{can} affect?'
This idea can be spelled out in various ways, but it seems to us that the most principled and plausible precisification is a \textit{causal domain restriction} \citep{bostrom2011infinite}, according to which an agent should evaluate the potential outcomes of her actions by applying the correct axiology only to those populations that might exist \textit{in her causal future} (presumably, her future light cone).\footnote{A causal domain restriction %on axiology 
might be motivated by the \textit{temporal value asymmetry}, our tendency to attach greater affective and evaluative weight to future events than to otherwise equivalent past events (\citeauthor{prior1959thank}, \citeyear{prior1959thank};	\citeauthor{parfit1984reasons}, \citeyear{parfit1984reasons}, Ch.\@ 8). It is sometimes claimed that this asymmetry characterizes only our self-regarding (and not our other-regarding) preferences (see e.g.\@ \citeauthor{parfit1984reasons}, \citeyear{parfit1984reasons}, p.\@ 181; \citeauthor{brink2011prospects}, \citeyear{brink2011prospects}, pp.\@ 378--9; \citeauthor{greene2015against}, \citeyear{greene2015against}, p.\@ 968; \citeauthor{dougherty2015future}, \citeyear{dougherty2015future}, p.\@ 3), but recent empirical studies appear to contradict this claim (\citeauthor{caruso2008wrinkle}, \citeyear{caruso2008wrinkle}; \citeauthor{greeneFChedonic}, forthcoming). However, though the temporal value asymmetry is a clear and robust \textit{psychological} phenomenon, it has proven notoriously difficult to come up with any normative \textit{justification} for asymmetric evaluation of past and future events (see for instance \cite{moller2002parfit}, \cite{hare2013time}).} %todo Double check these cites, think about whether they're the best cites to give. Consider also \cite{parfit1983rationality}, \cite{suhler2012thank}, DOUGHERTY, GREENE & SULLIVAN, SULLIVAN BOOK, YEHEZKEL, "THEORIES OF TIME AND THE ASYMMETRY IN HUMAN ATTITUDES"?
%To make this suggestion into a principled view, we cannot draw a  distinction between the part of the population that the agent  \textit{affects a lot} and the part that she \textit{affects only a  little}.  Rather, we should draw a line between the population that she  \textit{can in principle} affect and the population that she \textit{cannot  in principle} affect.  That is, we should adopt a \textit{causal domain  restriction} \citep{bostrom2011infinite}, applying our axiology only to the 
%population of an agent's causal future (presumably, her future light cone), 
%and not to the population of the world as a whole.  
Since background populations of the sort described in the last section will mostly lie outside an agent's future light cone, a causal domain restriction may drastically reduce the size of the population that can be treated as background, and hence the practical significance of %the results from \S \S \ref{section-AveragistViews}--\ref{section-EgalitarianViews}
our limit results.

Here are three replies to this suggestion. First, to adopt a causal domain restriction is to abandon a central and deeply appealing feature of consequentialism, namely, the idea that we have reason \textit{to make the world a better place}, from an impartial and universal point of view.
%once we adopt a causal domain restriction, we are no longer in the  business of doing \textit{axiology}, i.e., of considering the overall, agent-neutral goodness of worlds. Thus, 
That some act would make the world a better place, \textit{full stop}, is a straightforward and compelling reason to do it. It is much harder to explain why the fact that an act would make \textit{your future light cone} a better place (e.g., by maximizing the average welfare of its population), while making the world as a whole worse, should count in its favor.%
%[TT Sep 28] Adding refs from my separability chapter, also below (In general, not obvious how much detail we should go into here, since none of this is new)
\footnote{This point goes back to \citet{broad1914doctrine}; see \citet{carlson1995consequentialism} for a detailed discussion of this area.}

%To explain why an agent has reason  to bring about a world that is better according to some causal-domain-restricted criterion (e.g., that maximizes average welfare inside the agent's future light cone), we cannot appeal to the standard  motivation for consequentialism, namely, the reason-giving force of the  agent-neutral or impartial good.  We must explain why the agent has reason  to care about the overall value of her future light cone but not the world  as a whole. (And insofar as we are arguing for a non-separable axiology,  we cannot simply claim that the relative value of different outcomes inside  her future light cone is a reliable proxy for their relative value in the world as a whole.)
%This is a substantial theoretical cost.

Second, the combination of a causal domain restriction with a non-sepa\-rable axiology can generate counterintuitive inconsistencies between agents (and agent-stages) located at different times and places, with resulting inefficiencies. As a simple example, suppose that $A$ and $B$ are both agents who evaluate their options using causal-domain-restricted average utilitarianism. At $t_1$, $A$ must choose between a population of one individual with welfare 0 who will live from $t_1$ to $t_2$ (population $\Pop[X]$) or a population of one individual with welfare $-1$ who will live from $t_2$ to $t_3$ (population $\Pop[Y]$). At $t_2$, $B$ must choose between a population of three individuals with welfare 5 (population $\Pop[Z]$) or a population of one individual with welfare 6 (population $\Pop[W]$), both of which will live from $t_2$ to $t_3$. If $A$ chooses $\Pop[X]$, then $B$ will choose $\Pop[W]$ (yielding an average welfare of 6 in $B$'s future light cone), but if $A$ chooses $\Pop[Y]$, then $B$ will choose $\Pop[Z]$ (since $\Pop[Y + Z]$ yields average welfare $3.5$ in $B$'s future light cone, while $\Pop[Y + W]$ yields only $2.5$). Since $A$ prefers $\Pop[Y + Z]$ to $\Pop[X + W]$ (which yield averages of $3.5$ and $3$ respectively in $A$'s future light cone), $A$ will choose $Y$. Thus we get $\Pop[Y + Z]$, even though $\Pop[X + Z]$ would have been better from both $A$'s and $B$'s perspectives.\footnote{One general lesson of this example is that, when a group of timelike-related agents or agent-stages accept the same causal-domain-restricted non-separable axiology, an earlier agent in the group will have an incentive (i.e., will pay some welfare cost) to push axiologically significant events forward in time, into the future light cones of later agents, so that their evaluations of their options will more closely agree with hers.} That two agents who accept exactly the same normative theory and have exactly the same, perfect information can find themselves in such pointless squabbles is surely an unwelcome %"unseemly"?
feature of that normative theory, though we leave it to the reader to decide just how unwelcome.\footnote{
%[TT Sep 28] Not sure where to put this reference (or what exactly to say about it) but I guess it can go here; trimming the hurka ref to compensate
The argument is essentially due to \citet{Rabinowicz1989-RABAPD}; see also the cases of intertemporal conflict for future-biased average utilitarianism in \citet[pp.~118--9]{hurka1982more}.
%A similar case of intertemporal conflict generated by a future-biased form of average utilitarianism is described by \citeauthor{hurka1982more} (\citeyear{hurka1982more}, pp.\@ 118--9).
%[TT] TODO - double-check what we want to say about the citations
	
Of course, cases like these also create potential time-inconsistencies for individual agents, as well as conflict between multiple agents. But these inconsistencies might be avoidable by standard tools of diachronic rationality like `resolute choice'.} %todo* [cites?]

Third, a causal domain restriction might not be enough to avoid the limit behaviors described in \S\S \ref{section-AveragistViews}--\ref{section-EgalitarianViews}, if there are large populations inside our future light cones that are background 
%[TT Sep 28] Adding, since it's clearly necessary [CT] Not sure that it is. The formal rule for distinguishing "background" and "foreground" populations at the start of Section 6 was meant to have the consequence that most of the population inside your future light cone can contribute to the background rather than the foreground population, without any approximation (even if your choices affect their identities, for instance). That's only plausible if welfare levels are discrete (and maybe somewhat coarse-grained), though. And the formal background/foreground distinct at the top of Section 6 has the unfortunate consequence that you can't straightforwardly talk about particular individuals being background or foreground.
(at least, to a good approximation)
with respect to most real-world choice situations. For instance, it seems likely that most choices we face will have little effect on wild animal populations over the next 100 years. More precisely, our choices might be \textit{identity-affecting} with respect to many or most wild animals born in the next century (in the standard ways in which our choices are generally supposed to be identity-affecting with respect to most of the future population---see, e.g., \citeauthor{parfit1984reasons} (\citeyear{parfit1984reasons}, Ch.\@ 16)), %todo Should we give additional citations/a different citation here? Or just #CUT the parenthetical?
but will have little if any affect on the \textit{number} of individuals at each welfare level in that population. And this alone supplies quite a large background population---perhaps $10^{13}$ mammals and $10^{16}$ vertebrates. %As we will see, even this background population is enough to generate many of the practical conclusions associated with our limit results. 
Indeed, it is plausible that with respect to most choices (even comparatively major, impactful choices), the vast majority of the present and near-future \textit{human} population can be treated as background. For instance, if we are choosing between spending \$1 million on anti-malarial bednets or on efforts to mitigate long-term existential risks to human civilization, even the `short-termist' (bednet) intervention may have only  a comparatively tiny effect on the number of individuals at each welfare level in 
%...will only have a significant, predictable effect on the welfare distribution of---a tiny fraction of 
the present- and near-future human population, so that most of that population can be treated as background.\footnote{For further discussion of, and objections to, causal domain restrictions in the context of infinite ethics, see \cite{bostrom2011infinite} and \cite{arntzenius2014utilitarianism}.} %todo #CUT this footnote?

%todo #REINSERT? \footnote{There are complications here, however, that would require significant work to untangle. For instance, suppose that bednet distributions have a predictable effect on a tiny subpopulation, but have random effects on the rest of the present human population, which are symmetric in the sense that they have an equal chance of causing or preventing any given outcome. Can this `randomly affected' population be treated as background? Or suppose that these general-population are just slightly asymmetric, with some bias toward positive or toward negative effects? Exploring the behavior of non-additive axiologies in the context of these `background-ish' populations would require us to move to a risky/probabilistic setting that we have avoided in this paper.}

%todo ADD THIS BACK IN? (BUT IF SO, NOTE THAT EDT ALSO MAKES IT MORE COMPLICATED TO IDENTIFY A "BACKGROUND POPULATION") Fifth and finally, it is worth noting that the causal domain restriction move is hard to justify if we accept a non-causal decision theory. If we are evidential decision theorists, for instance, then we cannot claim that populations outside our future light cone are in any principled sense practically irrelevant, and so it seems all the more strange to ignore them in our assessment of outcomes.

\section{Objection 2: Counting some for less than one}
\label{section-CountingSomeLess}

Another way one might try to avoid the limit behaviors described in \S\S \ref{section-AveragistViews}--\ref{section-EgalitarianViews} is to claim that not all welfare subjects make the same contribution to the `size' of a population, as it should be measured for axiological purposes. 
%[TT Sep 28] Adding - I think the discussion is a bit WTF without an immediate example. 
Roughly speaking: although we should not deny \emph{tout court} that fish are welfare subjects, perhaps, when evaluating outcomes, a typical fish should effectively count as only (say) one tenth of a welfare subject, 
%[TT] VERY tempting to say `three fifths' but let's not [CT] That seems prudent...
given its cognitive and physiological simplicity. 
%[TT Sep 28] I'm simplifying this discussion, though I'm not sure it's wise; the original worry was: Is it true  that "the size of the background population relative to the foreground populations" is what matters? The more naive idea is simply that not counting fish for much makes the background population smaller, in absolute terms.
If, in a typical choice situation, the background population is predominantly made up of such simple creatures, then it might be dramatically smaller (in the relevant sense) than it would first appear.%
%More specifically, if in a typical choice situation the background population is predominantly made up of one type of welfare subject (say, fish or small mammals) while the possible foreground populations are predominantly made up of another type of welfare subject (say, humans), then stipulating that 
%%[TT Sep 28] adding, and make->makes
%each of the
%the latter makes a larger contribution to population size will (all else being equal) reduce the size of the background population relative to the foreground populations.
\footnote{Thanks to \ifanon [redacted] \else Tomi Francis \fi and \ifanon [redacted], \else Toby Ord, \fi who each separately suggested this objection.}

A bit more formally, we can understand this strategy as assigning a real-valued \textit{axiological weight} to each individual in a population, and turning populations from integer-valued to real-valued functions, where $\Pop[X](w)$ now represents not the \textit{number} of welfare subjects in $\Pop[X]$ with welfare $w$, but the \textit{sum of the axiological weights} of all the welfare subjects in $\Pop[X]$ with welfare $w$. 
%Another line of objection to the proceeding arguments claims that not all individuals/welfare subjects should be given the same weight from the perspective of axiology, i.e., in determining the overall value of a world. If we attach `axiological weights' to each individual in a population, then we would for instance determine the size of a population not by simply counting individuals, but rather by summing their axiological weights.
Axiological weights might be determined by factors like brain size, neuron count, lifespan, or by a combination of `spatial' and `temporal' factors (e.g., lifespan times neuron count). %todo COULD ALSO DO "NEURON COUNT IN PFC" OR SOMETHING LIKE THAT? WOULD THIS BE WELL-MOTIVATED?
%[TT] TODO I don't see the point of the "of `spatial' etc. clause" [CT] As in, when we say "could be determined by x, y, z factors", the "or a combination of those factrs" possibility is already implicit? I'm fine with cutting everything after "lifespan" (which was just meant for emphasis, as is possibly a little confusing since "neuron count" isn't exactly "spatial").
Weighting by lifespan seems particularly natural if we think that our ultimate objects of moral concern are \textit{stages}, rather than complete, temporally extended individuals. Weighting by brain size or neuron count may seem natural if we believe that, in some sense, morally significant properties like sentience `scale with' these measures of size.
%The proponent of axiological weights could claim that our discussion in in \S \ref{section-RealBackgroundPopulation} overstated the effective size of the background population relative to the present/near future human population, because the background population consists mainly of small animals that should receive less axiological weight than humans.

Here are three replies to this suggestion:
%Seems unintuitive/ad hoc/not otherwise motivated. 
First, of course, one might lodge straightforward ethical objections to axiological weights. They seem to contradict the ideals of impartiality and equal consideration that are often seen as central to ethics in general and axiology in particular (and for this reason, may be especially hard to reconcile with egalitarian views in axiology). It's also hard to imagine a plausible principle that assigns reduced axiological weight to non-human animals without also assigning reduced axiological weight to some humans, which many will find ethically unacceptable.

Second, the most natural measures by which we could assign %unequal 
axiological weights generate population size adjustments that, though large, still leave us with background populations significantly larger than the present human population.
%...generate adjustments to the size of the background  population that are not large enough to change our qualitative conclusions. 
For instance, suppose we stick with our conservative assumption that only 
mammals are welfare subjects, but also weight by cortical neuron count. 
And, very conservatively, let's take mice as representative of non-human 
mammals in general. Humans have roughly 2875 times as many cortical neurons 
as mice \citep[p.\@ 251]{roth2005evolution}. Normalizing our axiological weights so that present-day humans have an average weight of 1, this 
would mean that non-human mammals have an average weight of $3.48 \times 
10^{-4}$, which would cut our estimate of the size of the mammalian 
background population from $\sim 6.6 \times 10^{18}$ down to $\sim 2.3 
\times 10^{15}$. If we \textit{also} weight by lifespan, and generously 
assume that present-day humans have an average lifespan of 100 years, then 
the effective mammalian background population is reduced to $\sim 2.3 
\times 10^{13}$.\footnote{When we weight by lifespan, we can derive %an effective
population size simply from the number of individuals alive at a 
time multiplied by time, without needing to make any assumptions about  birth or death rates.} %todo Add another half sentence of explanation to this footnote?
Thus, even after making a host of conservative assumptions (only counting mammals as welfare subjects, taking a conservative estimate of the number of mammals alive at a time, ignoring times before the K-Pg boundary event, weighting by cortical neuron count and lifespan, and taking mice as a stand-in for all non-human mammals), we are still left with a background population more than three orders of magnitude larger than the present human population.

Third and finally, as we have already argued, even if we entirely ignore non-humans  %be enough to 
we may still find that background populations are large relative to foreground population in most present-day choice situations. To begin with, past humans outnumber present humans by more than an order of magnitude (as we saw in \S \ref{section-RealBackgroundPopulation}). And it seems plausible that the large majority even of the present and near-future human population is 
%[TT Sep 28] again
approximately
background 
%rather than foreground 
in most choice situations (as we argued at the end of \S \ref{section-CausalDomainRestriction}). Thus, even if we \textit{both} severely deprioritize or ignore non-humans \textit{and} adopt a causal domain restriction, we might \textit{still} find that background populations are usually large relative to foreground populations.

\section{The value of avoiding existential catastrophe}
\label{section-ExistentialCatastrophe}
%The value of avoiding existential catastrophe
%The value of preventing existential catastrophe
%The disvalue of existential catastrophe
%The value of survival
%The value of the future
%The value of the long-term future

%\textbf{[general hand-wavey thing about practical implications]}

Taking stock: in \S \S \ref{section-AveragistViews}--\ref{section-EgalitarianViews}, we showed that various non-additive axiologies converge to additive axiologies in the presence of large enough background populations. In \S \ref{section-RealBackgroundPopulation}, we argued that the background populations in real-world choice situations are very large---at least, multiple orders of magnitude larger than the affectable portion of the present and near-future population. And in \S \S \ref{section-CausalDomainRestriction}--\ref{section-CountingSomeLess}, we resisted two strategies for deflating the size of real-world background populations.

If we are right about the size of real-world background populations, this provides a weak \textit{prima facie} reason to believe that our limit results are %"practically relevant"? "practically informative"? "practically significant"? "significant in practice"
practically significant---i.e., that what is true in the limit will be true in practice, for the most plausible versions of the various families of non-additive axiologies we have considered. %That is, the fact that real-world background populations are very large in absolute terms and relative to the present population provides weak \textit{prima facie} evidence that they are large \textit{enough} for our limit results to `kick in', with plausible non-additive axiologies closely agreeing with their additive counterparts in real-world choice situations.
That is, the absolute and relative size of real-world background populations weakly suggests %This would mean..
that we should expect plausible non-additive axiologies to agree closely with their additive counterparts in real-world choice situations. More generally, it suggests that even if we don't accept (additive) separability as a fundamental axiological principle, it may nevertheless be a useful heuristic for real-world decision-making purposes---i.e., that arguments in practical ethics that rely on separability assumptions are likely to be truth-preserving in practice. %even if separability is not a basic ethical/axiological truth.

%even if we don't accept additivity as a fundamental axiological principle, it may nevertheless be a useful heuristic for real-world decision-making purposes.

%In general terms, if real-world background populations are large enough for plausible non-additive axiologies to exhibit their limit behavior, ...

%In general terms, it suggests that even if we don't accept additivity as a fundamental axiological principle, it may nevertheless be a useful heuristic for real-world decision-making purposes. More specifically, it suggests that any given non-additive axiology (from the families we have surveyed) may agree closely in practice with its additive counterpart.

%\textbf{[framing the question re existential catastrophe]}

\subsection{Present welfare vs. future population size}
%Think about how to divide things up into subsections in this section, and how to name the subsections. Maybe just get rid of subsection divisions entirely? 
%"Present welfare vs. future population size"
%"The welfare of the current generation vs. the existence of future generations"
%"Present welfare vs. long-term survival"
%"Modelling the tradeoff between present welfare and existential catastrophe"

But we will focus on a particular issue in practical ethics where we can say something a bit more concrete and definite. As we suggested in \S \ref{section-Introduction}, perhaps the most important practical implication of our results concerns the importance of existential catastrophes---more specifically, the extent to which the potentially astronomical scale of the far future makes it astronomically important to avoid existential catastrophe. An `existential catastrophe', for our purposes, is any near-future event that would drastically reduce the future population size of human-originating civilization %`any event that would greatly reduce the size of the future Earth-originating population of welfare subjects'.
(e.g., human extinction).\footnote{This is a broader category of events than `premature human extinction'---for instance, an event that prevented humanity from ever settling the stars, while allowing us to survive for a very long time on Earth, could be an existential catastrophe in our sense. It is also importantly distinct from the usual concept of %category of events usually labeled `existential catastrophe'
`existential catastrophe' in the philosophical literature, which is roughly `any event that would permanently curtail humanity's long-term potential for value' (see for instance \citeauthor{bostrom2013existential}, \citeyear{bostrom2013existential}, p.\@ 15; \citeauthor{ord2020precipice}, \citeyear{ord2020precipice}, p.\@ 37).}  To keep the discussion manageable, we will focus on $\AU$ and, secondarily, VV1/VV2. This lets us isolate %These axiologies exhibit 
the central relevant feature of insensitivity to scale (or asymptotic insensitivity to scale) in the absence of background populations, without the essentially orthogonal feature of inequality aversion.%
%[TT Sep 29] Adding clarification - this discussion feels rather jargony (perhaps we should more or less just say what's in the footnote, instead of the above
\footnote{For example, while totalist two-factor egalitarianism in not additive, it is relatively clear that it can give great value to avoiding existential catastrophe, since the value of a population scales with its size.}
%todo* Do more to explain this? Include something about intensivity/insensitivity to scale in the section on separability/additivity?
We will also focus on the case where the future generations that will exist if we avoid existential catastrophe have higher average welfare than the background population, so that $\AU$ assigns positive value to avoiding existential 
catastrophe,  at least in the large-background-population limit.
(But much of what we say about the value of avoiding existential catastrophe on this assumption also applies, \textit{mutatis mutandis}, to the \textit{dis}value of avoiding existential catastrophe on the opposite assumption that the potential future population has lower average welfare than the background population.) %todo How hard would it be to frame the discussion in this section in a  way that it's specific to the "extinction bad" case?

%Where the arguments in this section are specific to this case, roughly symmetrical things can be said about the \textit{dis}value of avoiding existential catastrophe in the case where the potential future population has lower average welfare than the background population.)

%[TT] Exposition feels too verbose - get to "Z+C+F vs Z+C'" faster.
The importance of avoiding existential catastrophe can be measured by comparing the value of avoiding existential catastrophe with the value of improving the welfare of the affectable pre-catastrophe population (which, for simplicity, we will hereafter call `the current generation'). 
We would like to know how the answer to this question depends on the welfare and (especially) the size of the background population.
%We would like to know how the  importance of avoiding existential catastrophe, in this sense, changes in the presence of 
%%realistic background populations. 
%large background populations. 

To formalize the question, let $\Pop[C]$ represent the current generation as it will be if we %optimize for present welfare  
prioritize its welfare at the expense of allowing an existential catastrophe. %---thus, it is relatively small but may have very high average welfare.
Let $\Pop[C']$ denote the current generation as it will be if we instead prioritize avoiding an existential catastrophe. Thus $\Avg[C] > \Avg[C']$, but we assume that $\Size[C] = \Size[C']$.  (This is mostly harmless:  it just means  that we designate as the members of $\Pop[C']$ the first $\Size[C]$ individuals in the affectable population in the world where we avoid existential catastrophe.) Let $\Pop[F]$ denote the future population that will exist only if we avoid existential catastrophe. And suppose there is a background population $\Pop[Z]$, which includes past terrestrial welfare subjects, perhaps distant aliens, and perhaps unaffectable present/future welfare subjects like wild animals.

%%%%%%%%%%%%%%%%%%
%%%%%%%%%%%%%%%%%%
%[TT Sep 29] I've rearranged the exposition here - it seems less waffly (and could profitably be made even more direct, I think). CT's previous version is in comments below
%%%%%%%%%%%%%%%%%%
%%%%%%%%%%%%%%%%%%

In short, we consider a choice between $Z+C$ and $Z+C'+F$.
In terms of this choice, the importance of avoiding existential catastrophe can be made precise in several different ways.
We will consider the following three: 
\begin{description}
\item[{\it Maximum incurred cost.}] Holding fixed the average welfare $\Avg[C]$ of the current generation in the world where existential catastrophe occurs, %and $\Pop[F]$ does not exist, 
what is the greatest reduction in welfare for the current generation that is worth accepting to avoid existential catastrophe? 
%[TT] i.e. holding fixed \Avg[C], how low can \Avg[C'] go before C'+F is worse than C?
\item[{\it Maximum opportunity cost.}] Holding fixed the average welfare $\Avg[C']$ of the current generation in the world where existential catastrophe \textit{does not} occur, % and $\Pop[F]$ exists,
what is the greatest improvement in the welfare of the current generation that is worth forgoing to to avoid existential catastrophe? 
%[TT] i.e. holding fixed \Avg[C'], how high can \Avg[C] go before C'+F is worse than C?
%(Naively, one might just ask `how big can \Avg[C]-\Avg[C'] be before C'+F is worse than C?' Perhaps though this doesn't admit as clean an answer?) [CT] Yeah, this is my impression at any rate. In the case where there's no background population, for instance, the answer is "\Avg[C]-\Avg[C'] can be really large as long as \Avg[C] < \Avg[F], but if \Avg[C] > \Avg[F], then you automatically prefer C to C' + F. (Possibly it's better to have a single measure, but separate the cases where \Avg[C] > \Avg[F] and \Avg[C] < \Avg[F]?
\item[{\it Value difference ratio.}] Holding fixed both $\Avg[C]$ and $\Avg[C']$, and thinking of $Z+C'$ as the status quo,   what is the ratio between the changes in value that would result from (i) avoiding existential catastrophe by adding $F$, versus (ii)~raising the welfare of the current generation from $\Avg[C']$ to $\Avg[C]$?
%[TT Sep 29] I found the previous formulation of the VDR impossible to understand, so I've thoroughly rewritten it
\end{description}

Broadly, we want to know how the presence of $\Pop[Z]$ affects these measures of importance. 
% i.e., affects the value of adding $\Pop[F]$ to the population as compared to the value of improving $\Pop[C']$ to $\Pop[C]$.
%[TT Sep 29] Sorry, this paragraph was hurting me so I'm making lots of random changes to try to find something I can understand, hopefully without destroying the content 
We know they depend, for one thing, on the size of $F$; we want particularly to know how this dependence is mediated by the size of $Z$.
%And more particularly, we want to know how the relationship between $\Size[F]$ and the value of adding $\Pop[F]$ to the population is mediated by $\Size[Z]$. 
In the extreme case, as $\Size[Z]\to\infty$, we know from our results in \S \ref{section-AveragistViews} that $\AU$, VV1, and VV2 all converge to $\CL_{\Avg[Z]}$. %in the large-background-population limit.
And according to $\CL_{\Avg[Z]}$, the value of adding $\Pop[F]$ to the population scales with $\Size[F]$, so that when $\Size[F]$ is astronomically large, the importance of avoiding existential catastrophe, by any of these measures, will be astronomically great. 
%[TT Sep 29] I had too hard a time understanding what `...linear relationship...hold to a good approximation' meant - we know that for any fixed Z, increasing F makes value approach a finite bound, and in that sense isn't approximately linear. So I'm trying to rephrase this in a way that makes more sense to me.
We should therefore expect, a bit roughly, that $\AU$ will give great importance to avoiding existential catastrophe when both $\Size[Z]$ and $\Size[F]$ are large, and more precisely that its measures of importance will agree with those of $\CL_{\Avg[Z]}$. The task is to say more about how this works at a qualitative level, and then (in \S\ref{subsection-illustration}) to give some indicative numerical results.

\subsection{Measure 1: Maximum incurred cost}

First, we hold fixed the welfare of the the current generation in the catastrophe world (where $\Pop[F]$ does not exist), and consider the greatest welfare cost we are willing to impose on the current generation to avoid catastrophe and thereby add $\Pop[F]$ to the population. 

According to the $\CL_{\Avg[Z]}$, the axiology to which $\AU$, VV1, and VV2 converge in the limit, this is simply the critical-level sum of welfare in $\Pop[F]$, given by $\Size[F](\Avg[F] - \Avg[Z])$. That is, when $\Tot(C) - \Tot(C') = \Size[F](\Avg[F] - \Avg[Z])$, $\CL$ is indifferent between $\Pop[Z+ C' + F]$ and $\Pop[Z+C]$.
%$V\sCL[{\Avg[Z]} ](\Pop[C' + F + Z] = V\sCL[{\Avg[Z]} ](\Pop[C + Z])$, so that $\CL$ is indifferent between these two composite populations.
%The maximum welfare cost we should be willing to impose on the current generation for the sake of avoiding existential catastrophe, according to $\AU$, 
%
%[TT Sep 29] I find it *so* much easier to parse when the populations are written in the order Z+C+F; I know in the rest of the paper we generally write the background population on the right, but the temptation to write F on the right is too great. So I'm standardizing this way in this section, nanny nanny poo poo
According to $\AU$, analogously, the maximum cost we are willing to impose on the current generation is the cost at which $\Avg[Z+C' + F ] = \Avg[Z+C ]$. We solve for it, %find an equation for it, 
therefore, by rearranging this equation into an equation for $\Tot(C) - \Tot(C')$ %, %(the welfare cost imposed on the current generation), 
%[TT Sep 29] I think this is clearer (and shorter!)
in terms of $Z$, $C$, and $F$.%
%where $\Avg[C]$ but not $\Avg[C']$ is allowed to appear on the other side of the equation.
\footnote{If we instead wanted to focus on the \textit{average} (per capita) welfare cost imposed on the current generation, we could just divide both sides of the following equation by $\Size[C]$.} This rearranged equation turns out to be:

\begin{equation}\label{eq:MaxIncurredCost}
\Tot(C) - \Tot(C') = \frac{\Size[Z]\Size[F](\Avg[F] - \Avg[Z]) + 
%[TT Sep 29] switching C and F - seems more understandable
\Size[C]\Size[F]
(\Avg[F] - \Avg[C])}{\Size[Z + C]}.
\end{equation}

%(Of course, from this equation we can also determine the maximum \textit{average}---i.e., per capita---welfare cost $\AU$ is willing to impose on the current generation to avoid existential catastrophe, by dividing both sides by $\Size[C]$.)

%The %fair price for preventing 
%maximum price we should be willing to pay to avoid existential catastrophe is the price---i.e., the value of $\Tot[C] - \Tot[C']$---at which $\Avg[C + Z] = \Avg[C' + F + Z]$. With a little rearrangement, we find that this price is given by

%$$\Tot[C] - \Tot[C'] = \Size[F](\Avg[F] - \Avg[C + Z]).$$
%[TT] Note our convention in othe sections is to write the background population last, i.e $C+Z$ rather than $Z+C$. Here it has gone the other way (which I instinctively prefer)  but it'd be nice to be consistent.

The key thing to notice about this equation is its surprising implication that the importance of avoiding existential catastrophe in the `maximum incurred cost' sense scales linearly with $\Size[F]$, with or without a background population. %...even in the absence of any background population. %...and is therefore potentially astronomical.
As we will see, this is not the case for the other two measures of the importance of avoiding existential catastrophe we consider. The right way to interpret this fact is as follows: if $\Avg[F] > \Avg[C]$ and $\Size[F] \gg \Size[C]$, then $\AU$ is willing to impose enormous costs on the current generation to enable the existence of $\Pop[F]$, since if $\Pop[F]$ exists, $\Pop[C']$ will be only a very small part of the resulting population and must have extremely low average welfare to reduce $\Avg[C' + F]$ below $\Avg[C]$. And on the other hand, if $\Avg[C] > \Avg[F]$ and $\Size[F] \gg \Size[C]$, then $\AU$ will require an enormous increase in the welfare of the current generation (i.e., that $\Avg[C'] \gg \Avg[C]$) to compensate for the reduction in average welfare created by $\Pop[F]$.

Nevertheless, even by this measure, the size of the background population makes a difference because it determines the `effective critical level' to which $\Avg[F]$ is compared---the average welfare level above which adding $\Pop[F]$ to the population has positive value, and below which it has negative value. When $\Size[C] \gg \Size[Z]$, the right-hand side of \eqref{eq:MaxIncurredCost} is approximately 
$\Size[F](\Avg[F] - \Avg[C])$;\footnote{
%[TT Sep 29] adding explanation; I suppose there's something slightly funny going on when y=0 but it takes care of itself in the end.
Formally, `if $a\gg b$ then $x$ is approximately $y$' means that $\lim_{a/b\to\infty} x/y=1$. In this case, the limit converges uniformly in $\Size[F]$. 
}
thus $\AU$ agrees closely with $\CL_{\Avg[C]}$ rather than $\CL_{\Avg[Z]}$ and is 
%[TT Sep 29] moving the "only" for better sense
only
willing to impose any positive cost at all on the current generation to avoid existential catastrophe 
%only 
 when (with some approximation) $\Avg[F] > \Avg[C]$.
%Describe this in terms of limit behavior? "(As $\Size[Z] \to 0$, $\Tot[C] - \Tot[C'] \to \Size[F](\Avg[F] - \Avg[C])$, as per Eq.\@ \ref{eq:MaxIncurredCost}.)" 
But when $\Size[Z] \gg \Size[C]$, the right-hand side of \eqref{eq:MaxIncurredCost} is approximately $\Size[F](\Avg[F] - \Avg[Z])$---i.e., the value given by $\CL_{\Avg[Z]}$. %in agreement with $\CL_{\Avg[Z]}$. %in line with our limit results.
%Describe this in terms of limit behavior? "(As $\Size[Z] \to \infty$, $\Tot[C] - \Tot[C'] \to \Size[F](\Avg[F] - \Avg[Z])$.)"
This shift could either increase or decrease the value of avoiding existential catastrophe (depending on whether $\Avg[Z]$ is greater than or less than $\Avg[C]$), and could reverse the sign of the value of avoiding existential catastrophe if $\Avg[F]$ is between $\Avg[Z]$ and $\Avg[C]$. 
%[TT] Not clear `most notably' is true
Most notably for our purposes, 
the effective critical level will be closer to $\Avg[Z]$ than to $\Avg[C]$ if %as long as 
$\Size[Z] > \Size[C]$, 
%[TT] didn't we already say this - 
and will be very close to $\Avg[Z]$ if $\Size[Z] \gg \Size[C]$ (since $\Size[Z]\Size[F](\Avg[F] - \Avg[Z])$ rather than $\Size[F]\Size[C](\Avg[F] - \Avg[C])$ will dominate the numerator in \eqref{eq:MaxIncurredCost}). %todo Say something like "So this starts to answer the 'When do the limit results kick in?' question"?
So by this measure, $\AU$ closely agrees with its corresponding additive limit theory  as long as the background population is substantially larger than the current generation, i.e., $\Size[Z] \gg \Size[C]$.

%%%%%%%%%%%%%%%%%%%%%%%%%%%%%
%%%%%%%%%%%%%%%%%%%
%[TT] My alternative version of the last paragraph, somewhat less informative
%%%%%%%%%%%%%%%%%%%
%Nevertheless, even by this measure, the size of the background population makes a difference because it determines the `effective critical level' to which $\Avg[F]$ is compared---the average welfare level above which adding $\Pop[F]$ to the population has positive value, and below which it has negative value. In the limiting case where $|Z|\to 0$, the right-hand side of \eqref{eq:MaxIncurredCost} is just
%$\Size[C]\Size[F](\Avg[F]-\Avg[C])$,
%corresponding to a critical level of $\Avg[C]$; in the other limit, where $|Z|\to \infty$, 
%the maximum incurred cost is instead 
%$\Size[F](\Avg[F]-\Avg[Z])$, 
%corresponding to a critical level $\Avg[Z]$ (as, of course, we expect from Theorem \ref{t:AU}). 
%%%%%%%%%%%%%%%%%%%%%%%%%%%%%
%%%%%%%%%%%%%%%%%%%%%%%%%%%%%

\subsection{Measure 2: Maximum opportunity cost}

Now let's ask the converse question: holding fixed the welfare of the current generation in the world \textit{without} existential catastrophe
%[TT Sep 29] adding because I find it hard to keep track of these things - if we have terminology we should use it!
(i.e. holding fixed $\Avg[C']$)%
, how large a welfare \textit{gain} for the current generation should we be willing to \textit{forgo} to avoid existential catastrophe? 

Here again, $\CL$ gives the answer $\Size[F](\Avg[F] - \Avg[Z])$. %To answer the question for $\AU$, 
To find $\AU$'s answer, we rearrange $\Avg[Z+C' + F ] = \Avg[Z+C]$ into an equation for $\Tot(C) - \Tot(C')$, this time 
%[TT Sep 29] Again, seems better to think "C is variable, so solve in terms of C'"
in terms of $Z$, $F$, and $C'$.
%allowing $\Avg[C']$ but not $\Avg[C]$ to appear on the other side of the equation. 
This gives us:
%[TT Sep 29] So putting C' everywhere (i.e. in the sizes) on the right-hand side. Also reversing F and C' in second term
\begin{equation}\label{eq:MaxOpportunityCost}
\Tot(C) - \Tot(C') = \frac{\Size[Z]\Size[F](\Avg[F] - \Avg[Z]) + \Size[C']\Size[F](\Avg[F] - \Avg[C'])}{\Size[Z + C' + F]}.    
\end{equation}

%[TT Sep 29] unclear "since"; also, just say C
Now the size of the background population takes on greater significance. Consider three cases:
%Now, since we are allowing the welfare of the current generation in the world without $\Pop[F]$ to vary, the size of the background population takes on greater significance. Consider three cases:

\begin{description}
%[TT] It's a bit unclear what this means - in what order do you take the limits? 
\item[Case 1:] $\mathbf{\Size[F] \gg \Size[C'] \gg \Size[Z].}\:$ 
In this case, the 
%[TT Sep 29]  here and below
right-hand side
%rhs 
of \eqref{eq:MaxOpportunityCost} is   approximately $\Size[C'](\Avg[F] - \Avg[C'])$, and the value of avoiding existential catastrophe as measured by maximum opportunity cost is therefore approximately independent of $\Size[F]$.\footnote{
%[TT Sep 29] adding explanation
Formally, a claim to the effect of `if $a\gg b\gg c$ then $x$ is approximately $y$' means that
$x/y\to 1$ as $a/b$ and $b/c\to \infty$; more precisely, for any $\epsilon>0$, there exists $n>0$ such that if both $a/b$ and $b/c$ are bigger than $n$, then  $x/y\in(1-\epsilon,1+\epsilon)$.}
%As $\Size[Z] \to 0$, the rhs of Eq.\@ \ref{eq:MaxOpportunityCost} converges to $\frac{\Size[F]\Size[C](\Avg[F] - \Avg[C'])}{\Size[F + C]}$, which (since $\Size[F] \gg \Size[C]$) is approximately $\Size[C](\Avg[F] - \Avg[C'])$, and therefore (approximately) independent of $\Size[F]$.
\item[Case 2:] $\mathbf{\Size[F] \gg \Size[Z] \gg \Size[C'].}\:$ In this case, the right-hand side of \eqref{eq:MaxOpportunityCost} is approximately $\Size[Z](\Avg[F] - \Avg[Z])$. Thus the value of avoiding existential catastrophe as measured by maximum opportunity cost is approximately proportional to $\Size[Z]$, which may be astronomically large but is also (we are supposing) much less than $\Size[F]$. Note also that the effective critical level is now close to $\Avg[Z]$ rather than $\Avg[C']$ as in Case 1. %shift from $(\Avg[F] - \Avg[C'])$ to $(\Avg[F] - \Avg[Z])$, indicating a shift in the `effective critical level'
\item[Case 3:] $\mathbf{\Size[Z] \gg \Size[F] \gg \Size[C'].}\:$ In this case, the right-hand side of \eqref{eq:MaxOpportunityCost} is approximately $\Size[F](\Avg[F] - \Avg[Z])$, in agreement with $\CL_{\Avg[Z]}$. Thus the value of avoiding existential catastrophe as measured by maximum opportunity cost is approximately proportional to $\Size[F]$, and will be astronomically large if $\Size[F]$ is astronomically large and $(\Avg[F] - \Avg[Z])$ is non-trivial.
\end{description}

%[TT Sep 29] I feel like the three cases really take a lot of commitment to absorb, and I think we should have more easy-to-remember stuff like this:
While there are a number of points of interest in this analysis, the quick takeaway is that the maximum opportunity cost increases without bound as we increase \emph{both} $\Size[F]$ and $\Size[Z]$ (while holding all else equal)---a situation reflected in Cases 2 and 3 but not Case 1. So, qualitatively, arguments from astronomical scale can go through if we attend to the potentially astronomical scale of both the future population \emph{and} the background population.

\subsection{Measure 3: Value difference ratio}

%[TT Sep 29] Redescribing because it hurts my head (probably this is just a matter of competing tastes...) Also displaying formula and giving it a name
Finally, we treat $Z+C'$ as a baseline, and ask whether it is better  to avoid existential catastrophe by adding $F$
or to improve $C'$ to $C$ . More precisely, we consider the ratio of the value of these improvements: 
%Finally, suppose we hold fixed both $\Avg[C]$ and $\Avg[C']$ and consider the \textit{ratio} between the improvements in the overall value of the population brought about by, on the one hand, improving $\Pop[C']$ to $\Pop[C]$ or, on the other hand, adding $\Pop[F]$ (by avoiding existential catastrophe).  That is, we consider the ratio 
\[
    R=\frac{V(\Pop[Z+C'+F ]) - V(\Pop[Z+C'])}{V(\Pop[Z+C ]) - V(\Pop[Z+C'])}.
\] 
According to $\CL_{\Avg[Z]}$, $R$ is equal to $\frac{\Size[F](\Avg[F] - \Avg[Z])}{\Size[C](\Avg[C] - \Avg[C'])}$. According to $\AU$, of course, $R$ is equal to $\frac{\Avg[Z+C'+F] - \Avg[Z+C']}{\Avg[Z+C] - \Avg[Z+C' ]}$. But again, we need to do some rearranging to make clear %highlight? make perspicuous?
how this ratio is affected by the 
sizes of $\Pop[Z]$, $\Pop[C]$, and $\Pop[F]$. Specifically, %$\frac{\Avg[F + C' + Z] - \Avg[C' + Z]}{\Avg[C + Z] - \Avg[C' + Z]}$ 
%[TT Sep 29] adding
in the case of $\AU$, the formula for $R$
%it 
rearranges to
\begin{equation}\label{eq:ValueRatio}
    \frac{1}{\Avg[C] - \Avg[C']} \left( \Avg[F]\frac{\Size[F]\Size[Z + C]}{\Size[C]\Size[Z + C + F]} - \Avg[C']\frac{\Size[F]}{\Size[Z + C + F]} - \Avg[Z]\frac{\Size[Z]\Size[F]}{\Size[C]\Size[Z + C + F]} \right).
\end{equation}

This expression is unattractive, %hideous, 
but informative. %illuminating
Again, let's consider three cases:

\begin{description}
\item[Case 1:] $\mathbf{\Size[F] \gg \Size[C] \gg \Size[Z].}\:$ 
In this case, \eqref{eq:ValueRatio} 
is approximately $\frac{\Avg[F] - \Avg[C']}{\Avg[C] - \Avg[C']}$, and the importance of avoiding existential catastrophe by the value difference ratio measure is therefore approximately independent of $\Size[F]$.
%As $\Size[Z] \to 0$, the rhs of Eq.\@ \ref{eq:MaxOpportunityCost} converges to $\frac{\Size[F]\Size[C](\Avg[F] - \Avg[C'])}{\Size[F + C]}$, which (since $\Size[F] \gg \Size[C]$) is approximately $\Size[C](\Avg[F] - \Avg[C'])$, and therefore (approximately) independent of $\Size[F]$.
\item[Case 2:] $\mathbf{\Size[F] \gg \Size[Z] \gg \Size[C].}\:$ In this case, \eqref{eq:ValueRatio} is  approximately $\frac{\Size[Z]}{\Size[C]} \times \frac{\Avg[F] - \Avg[Z]}{\Avg[C] - \Avg[C']}$. Thus the importance of avoiding existential catastrophe by the value difference ratio measure is approximately proportional to $\frac{\Size[Z]}{\Size[C]}$. %, which may be astronomically large but is also (we are supposing) much less than $\Size[F]$. 
And again, note that when $\Size[Z] \gg \Size[C]$, the effective critical level is close to $\sim \Avg[Z]$ rather than $\Avg[C']$.
\item[Case 3:] $\mathbf{\Size[Z] \gg \Size[F] \gg \Size[C].}\:$
In this case, \eqref{eq:ValueRatio} is approximately $\frac{\Size[F]}{\Size[C]} \times \frac{\Avg[F] - \Avg[Z]}{\Avg[C] - \Avg[C']}$, in agreement with $\CL_{\Avg[Z]}$. Thus the importance of avoiding existential catastrophe by the value difference ratio measure is now approximately proportional to $\frac{\Size[F]}{\Size[C]}$, and will be astronomically large if $\frac{\Size[F]}{\Size[C]}$ is astronomically large and $\frac{\Avg[F] - \Avg[Z]}{\Avg[C] - \Avg[C']}$ is non-trivial.
\end{description}

%[TT Sep 29] Keeping this bit of my ALT
As with the maximum opportunity cost, the most basic qualitative point is that the value difference ratio $R$ increases without bound as we increase \emph{both} $\Size[F]$ and $\Size[Z]$. The fact that possible future and actual background populations are both likely to be extremely large suggests that the value difference ratio will be greater than $1$ (thus favouring extinction-avoidance) for a robust range of the other parameters. 

%%%%%%%%%%%%%%%%%%%%%%%%%%%
%%%%%%%%%%%%%%%%%%%%%%%%%%%
%[TT] My alt
%%%%%%%%%%%%%%%%%%%%%%%%%%%
%%%%%%%%%%%%%%%%%%%%%%%%%%%
%This expression is unattractive, %hideous, 
%but informative. %illuminating
%Again, let's start by considering the small- and large-background limits. If we take the limit $\Size[Z]\to 0$ and \emph{then} the limit $\Size[F]\to\infty$,  $R$ approaches the finite bound
%$\frac{\Avg[F] - \Avg[C']}{\Avg[C] - \Avg[C']}$, corresponding to an effective critical level of $\Avg[C']$. In contrast, in the limit $\Size[Z]\to \infty$, $R$ approaches the value given to it by $\CL$, with critical level $\Avg[Z]$; it is linear in $\Size[F]$.  
%
%%[TT Sep 29] Not sure we really needed the ugly formula for this level of analysis! Oh well.
%
%As with the maximum opportunity cost, however,  perhaps the most important qualitative point is that the value difference ratio $R$ increases without bound as we increase \emph{both} $\Size[F]$ and $\Size[Z]$. The fact that possible future and actual background populations are both likely to be extremely large suggests that the value difference ratio will be greater than $1$ (thus favouring extinction-avoidance) for a robust range of parameters. 
%

\subsection{Illustration}\label{subsection-illustration}

%[TT Sep 29] I think I still like the alternative intro I wrote, so I'm tentatively keeping that (perhaps the addition I made at the end of the last subsection covers what's lost) [CT Sept 30] Agree, revised intro feels punchier.

So far, our analysis has remained qualitative; we'll now put in some numbers, with the purpose of illustrating two things: first, the practical point that even $\AU$ will give great weight to avoiding existential catastrophes, for some reasonable and even conservative estimates of the background population and other parameters; second, the more theoretical point that $\AU$ converges to $\CL$ with high precision,  given these same estimates.

%%%%%%%%%%%%%%%%%%%%%%%%%%%%%
%%%%%%%%%%%%%%%%%%%%%%%%%%%%%
%[TT Sep 29] CT's previous version
%%%%%%%%%%%%%%%%%%%%%%%%%%%%%
%%%%%%%%%%%%%%%%%%%%%%%%%%%%%
%We have argued for various approximations of the importance of avoiding existential catastrophe, according to $\AU$, when there are large size differences among the relevant populations (when $\Size[F] \gg \Size[C] \gg \Size[Z]$, $\Size[F] \gg \Size[Z] \gg \Size[C]$, or $\Size[Z] \gg \Size[F] \gg \Size[C]$). And most importantly, we have argued that when $\Size[Z] \gg \Size[F] \gg \Size[C]$, $\AU$ will agree closely with its additive limit theory $\CL_{\Avg[Z]}$, by all three of the measures we have considered. Let's now give a numerical illustration of these arguments, to see just how close these approximations really are, under realistic/conservative assumptions about the various populations.

For the sizes of the foreground populations, let's suppose that $\Size[C] = \Size[C'] = 10^{10}$ (a realistic estimate of the size of the present and near-future human population) and $\Size[F] = 10^{17}$ (a fairly conservative estimate of the potential size of the future human population, if we avoid existential catastrophe, arrived at by assuming $10^{10}$ individuals per century for the next billion years). %todo* Say more here? Make clear that you get to this by imagining that civilization continues on Earth until Earth becomes uninhabitable, but never spreads to the stars? And note that you could get the same number by assuming $2 \times 10^{10}$ individuals per century for 500 million years?
%Maybe lop off an order of magnitude here, so that all the population size differences are 3 oom? E.g., "2 billion people per century for 500 million years"? ("2 billion people per century" could happen with present population sizes but increased individual longevity.) 
%This isn't accounting for the affect that existential catastrophe would have on future non-human terrestrial population sizes, but probably fine to ignore that.
For $\Size[Z]$, we will consider three values: $\Size[Z] = 0$ (i.e., the absence of any background population), $\Size[Z] = 10^{13}$ (a rounding-down of our most conservative estimate of the number of past mammals, weighted by lifespan and cortical neuron count, from \S \ref{section-CountingSomeLess}), and $\Size[Z] = 10^{3}\times \Size[F] = 10^{20}$ (arrived at by assuming that the universe contains 1000 other advanced civilizations, of the same scale that our civilization will achieve if we avoid existential catastrophe).

In terms of average welfare, we have much less to go on. 
For simplicity let's assume that $\Avg[F] = 2$ 
%[TT Sep 29] I think we have to say something about normalization; not sure what exactly is best to say, but adding:
(corresponding to very good but generally normal human lives)
and $\Avg[Z] = 0$ (plausible for the case where $\Pop[Z]$ consists mainly of wild animals, somewhat less plausible for the case where it consists mainly of the member of other advanced civilizations). And let's assume that $\Avg[C'] = 1$ (except when considering maximum incurred cost, where $\Avg[C']$ is a dependent variable) and $\Avg[C] = 1.5$ (except when considering maximum opportunity cost, where $\Avg[C]$ is a dependent variable).

\begin{table}\small
	\begin{center}
		\begin{tabular}{lllll}
			\toprule
		 	\textbf{Axiology} & $\mathbf{\Size[Z]}$ & \textbf{MIC} & \textbf{MOC} & \textbf{VDR} \\ \midrule
			$\AU$ & $\Size[Z] = 0$  & $5 \times 10^{16}$ & $\sim 10^{10}$ & $\sim 2$ \\ 
			$\AU$ & $\Size[Z] = 10^{13}$  & $\sim 1.9985 \times 10^{17}$ & $\sim 2.0008 \times 10^{13}$ & $\sim 4.0016 \times 10^3$  \\ 
			$\AU$ & $\Size[Z] = 10^{20}$ & 
		%[TT] TODO do we need more sig figs here, below? Above, replacing 10^{3}|F| with 10^{20}
			$\sim 2 \times 10^{17}$ & $\sim 1.998 \times 10^{17}$ & $\sim 3.996 \times 10^7$ \\
			$\CL$ & --- & $2 \times 10^{17}$ & $2 \times 10^{17}$ &  %\frac{2 \times 10^{17}}{5 \times 10^9} = ...
			$4 \times 10^7$ \\
			\bottomrule
		\end{tabular} 
		\caption{The importance of avoiding existential catastrophe, as measured by maximum incurred cost (MIC), maximum opportunity cost (MOC), and value difference ratio (VDR), according to $\AU$ for different background population sizes and $\CL_{\Avg[Z]}$, with $\Avg[F] = 2$, $\Size[F] = 10^{17}$, $\Avg[C] = 1.5$, $\Avg[C'] = 1$, $\Size[C] = \Size[C'] = 10^{10}$, $\Avg[Z] = 0$, and $\Size[Z]$ as specified in each row.} %todo* [CT] I guess there should be units in the MIC and MOC columns... Is there a standard unit of welfare?? A welfare-on??
		\label{tbl:existentialCatastrophe}
	\end{center}	
\end{table}

Table \ref{tbl:existentialCatastrophe} gives the importance of avoiding existential catastrophe according to $\AU$ and $\CL_{\Avg[Z]}$, under these assumptions, for all three measures of importance and all three background population sizes. In general, we see that with three- or four-order-of-magnitude differences in the population sizes of $\Pop[C]$, $\Pop[F]$, and $\Pop[Z]$, the approximations arrived at above are accurate to at least the third or fourth significant figure. And more specifically, in the case where $\Size[Z] \gg \Size[F] \gg \Size[C]$, $\AU$ agrees with $\CL_{\Avg[Z]}$ on all three measures to at least the fourth significant figure.
%[TT] TODO - do we need more sig figs under MIC to back this up? [CT] Backstory: Mathematica wouldn't give me any more significant figures. :-) Possibly reducing |F| by an order of magnitude would help? (And would give us consistent 3-oom differences.)

\subsection{Conclusions}
\label{section-ExistentialCatastropheConclusions}

In summary: when the background population is small or non-existent, the importance of avoiding existential catastrophe according to $\AU$ is approxi\-mate\-ly proportional to $\Avg[F] - \Avg[C']$ or $\Avg[F] - \Avg[C]$ (depending on which measure we consider), and approximately independent of population size,
%the relative sizes of the $\Pop[F]$, $\Pop[C]$, and $\Pop[C]$
and is therefore unlikely to be astronomically large. When the background population is much larger than the current generation, but still much smaller than the potential future population, the importance of avoiding existential catastrophe according to $\AU$ approximately scales with $\Size[Z]$, and may therefore be astronomically large, while still falling well short of its importance according to $\CL_{\Avg[Z]}$. %its critical-level value $\Size[F](\Avg[F] - \Avg[Z])$. 
Finally, if the background population is much larger even than the potential future population (as it would be, for instance, if it includes many advanced civilizations elsewhere in the universe), $\AU$ agrees closely with $\CL[{\Avg[Z]} ]$ about the importance of avoiding existential catastrophe, treating it as approximately linear in $\Size[F]$, by all three of the measures we considered.
%the importance of avoiding existential catastrophe according to $\AU$ [approximately] scales with $\Size[F]$, and closely approximates its importance according to $\CL_{\Avg[Z]}$. 
The exception to this pattern is the `maximum incurred cost' measure, by which the importance of avoiding existential catastrophe scales with $\Size[F]$ regardless of the size of the background population.

%[TT] TODO read these three paragraphs more carefully
In this very specific context, therefore, we can now say how large the background population needs to be for large-background-population limiting behavior to `kick in': $\AU$ closely approximates $\CL_{\Avg[Z]}$ in every respect we have considered only when $\Size[Z] \gg \Size[F]$ 
%[TT] TODO this is way too cryptic
(or at any rate, only when $\Size[Z] > \Size[F]$). But it behaves in important ways like $\CL_{\Avg[Z]}$ as long as $\Size[Z] \gg \Size[C]$---both in that it is disposed to assign astronomical importance to avoiding existential catastrophes, %disposed to assign astronomical significance to existential catastrophes and other astronomically large changes in population size, 
and in that the effective critical level that determines whether that importance is positive or negative is approximately $\Avg[Z]$. This lends significance to our conclusion in \S \ref{section-RealBackgroundPopulation} that real-world background populations are much larger than the current generation (i.e., the affectable present and near-future population), whether or not they are large relative to the potential future population as a whole. The former fact alone is enough to have a significant effect on how $\AU$ evaluates existential catastrophes in practice. %to make the background population practically significant with respect to the evaluation of existential catastrophes under $\AU$.

Our conclusions about $\AU$ also partially generalize to VV1 and VV2. In the case of VV1: for any two populations $\Pop[X]$ and $\Pop[Y]$, if $\Avg[X] > \Avg[Y]$, $\Size[X] \geq \Size[Y]$, and $\Avg[X] \geq 0$, then clearly any VV1 axiology will prefer $\Pop[X]$ to $\Pop[Y]$. For our purposes, this means that any VV1 axiology, so long as it assigns non-negative value to the non-catastrophe population $\Pop[Z+C' + F]$ (i.e., so long as $\Avg[Z+C' + F] \geq 0$), will prefer that population to the catastrophe population $\Pop[C + Z]$ whenever $\AU$ does. Analogously, in the case of VV2 (which, recall, applies an increasing transformation $f$ to the average welfare of a population): for any two populations $\Pop[X]$ and $\Pop[Y]$, if $\Avg[X] > \Avg[Y]$, $\Size[X] \geq \Size[Y]$, and $f(\Avg[X]) \geq 0$, then clearly any VV2 axiology will prefer $\Pop[X]$ to $\Pop[Y]$. For our purposes, this means that any VV2 axiology, so long as it assigns non-negative value to the non-catastrophe population $\Pop[Z+C' + F ]$ (i.e., so long as $f(\Avg[Z+C' + F ]) \geq 0$), will prefer that population to the catastrophe population $\Pop[Z + C]$ whenever $\AU$ does.

Putting these observations together, any VV axiology, as long as it assigns positive value to the non-catastrophe population, will prefer it to the catastrophe population whenever $\AU$ does. This means, among other things, that under this condition, the importance of avoiding existential catastrophe as measured by maximum incurred cost or maximum opportunity cost, will be at least as great according to VV as according to $\AU$.\footnote{Consider VV2, of which VV1 is a special case (where $f(\Avg[X]) = \Avg[X]$). If $f(\Avg[Z+C'+F + C' ]) = f(\Avg[Z+C ])$, and is positive, then $g(\Size[Z+C'+F ])f(\Avg[Z+C'+F]) > g(\Size[Z+C])f(\Avg[Z+C])$, since $g$ is increasing.
%[TT] TODO - weakly/strongly problem
Thus, all else being equal, VV2 axiologies will require either a larger value of $\Avg[C]$ or a smaller value of $\Avg[C']$ to equalize the value of the populations, meaning that the maximum incurred cost/maximum opportunity cost that it will accept to avoid existential catastrophe is greater.

This does not necessarily mean that VV will converge with $\CL_{\Avg[Z]}$ faster than $\AU$, with respect to these measures, as the size of the background population increases. After all, $g$ may be arbitrarily close to linear up to arbitrarily large population sizes, allowing VV to remain in close agreement with $\TU$ rather than $\CL_{\Avg[Z]}$ for arbitrarily large populations. But it does mean that VV will converge with $\CL_{\Avg[Z]}$ faster than $\AU$ if it is converging from below.} With respect to value difference ratio, things are a bit more complicated: when $\Avg[Z+C'+F] \geq 0$ \textit{and} $\Avg[Z+C'+F] \geq \Avg[Z+C']$, VV1 is guaranteed to assign more importance than $\AU$ to avoiding existential catastrophe by this measure. But we cannot say anything analogous about VV2 in this case, since the transformation $f$ it applies to average welfare can be arbitrarily concave or convex.\footnote{If $\Avg[Z+C'+F]$ and $\Avg[Z+C'+F] - \Avg[Z+C']$ are both non-negative, then $\frac{g(\Size[Z+C+F])\Avg[Z+C'+F] - g(\Size[Z+C])\Avg[Z+C']}{g(\Size[Z+C ])\Avg[Z+C] - g(\Size[Z+C])\Avg[Z+C']} \geq \frac{\Avg[Z+C'+F] - \Avg[Z+C']}{\Avg[Z+C ] - \Avg[Z+C']}$. (Again, this means that VV1 will converge with $\CL_{\Avg[Z]}$ faster than $\AU$, with respect to the value difference ratio measure, if it is converging from below.) However, since VV2's $f$ need only be increasing, $\frac{f(\Avg[Z+C' + F]) - f(\Avg[Z+C'])}{f(\Avg[Z+C]) - f(\Avg[Z+C'])}$ can differ to an arbitrarily extreme degree from $\frac{\Avg[Z+C'+F] - \Avg[Z+C']}{\Avg[Z+C] - \Avg[Z+C']}$ (except when $\Avg[Z+C'+F] = \Avg[Z+C]$).}

%If $\Avg[Z] \geq 0$, then $V\sCL[{\Avg[Z]} ](S + Z) > V\sCL[{\Avg[Z]} ](C + Z)$ implies $V\sTU(S + Z) > V\sTU(C + Z)$. This means that VV1 will 

%\textbf{[Qualifications and caveats]}

%There are several complications and caveats to these conclusion that should be immediately noted. 

%Second, some non-additive axiologies already exhibit approximately linear sensitivity to differences in population size, in the absence of background populations---in particular, among the views we have considered, the family of totalist two-factor egalitarian views described in \S \ref{section-twoFactorEgalitarianism}. Thus, our results particularly serve to revise the \textit{prima facie} implications regarding the significance of existential catastrophe of those non-additive axiologies that regard changes in population size as having either no significance (AU, averagist two-factor egalitarian views) or diminishing marginal significance (VV1, VV2, RDU).

A crucial limitation of our discussion, however, is that we have only considered the \textit{objective importance of existential catastrophes}, %, according to non-additive axiologies in the presence of large background populations. 
and not the prospective or decision-theoretic significance of existential \textit{risks} (i.e., risks of existential catastrophe). If we assume a straightforward expectational decision theory according to which average utilitarians, for instance, should simply maximize expected average welfare, %under risk,
%If the non-additive axiologies we have considered are combined with a simple expectational theory of decision-making under risk (i.e., `maximize the expectation of $V_\Ax$'), 
then the astronomical decision-theoretic significance of existential risk would follow straightforwardly from the astronomical axiological significance of existential catastrophe in the `value difference ratio' sense (assuming, of course, that we can have non-negligible effects on the probability of existential catastrophe). We have sidestepped the question of risk, %questions of risk and uncertainty,
%or, "leave questions of risk and uncertainty for future research"
however, because there are good reasons to think that non-additive axiologists should be in the market for something other than this simple expectational theory of decision-making under risk\footnote{See for instance 
\citet[ch.~3]{thomas2016topics},
\citet[Prop.~4.8]{mccarthy2020utilitarianism},
\citeauthor{nebelFCrank} (forthcoming), \ifanon [redacted]\else \citeauthor{tarsneyMSaverage} (unpublished)\fi.}, and there is not yet any unproblematic or widely accepted alternative in the literature. %i.e., theory of decision-making under risk for non-additive axiologies like AU or RDU.
We therefore leave the question of how $\AU$, VV1, VV2, and other non-additive axiologies evaluate existential risk in the presence of large background populations for future research.

\section{Other implications}
\label{section-PracticalImplcations}

We conclude %our discussion 
by briefly surveying three other interesting implications of our limit results and, more generally, of the influence of background populations on the preferences of non-separable axiologies.

\subsection{Repugnant Addition}
\label{section-RepugnantAddition}

The Repugnant Conclusion, recall, is the conclusion (implied by TU among other axiologies) that for any positive welfare levels $l_1 < l_2$ and any number $n$, there is a population where everyone has welfare $l_1$ that is better than a population of $n$ individuals all with welfare $l_2$. One of the motivations for population axiologies with an `averagist' flavor (like AU, VV1, VV2, and QAA) is to avoid the Repugnant Conclusion. But the results in \S\S \ref{section-AveragistViews}--\ref{section-EgalitarianViews} imply that, although they avoid the Repugnant Conclusion as stated above, these views cannot avoid the closely related phenomenon of `Repugnant Addition': for any positive welfare levels $l_1 < l_2$ and any number $n$, if $\Pop[Y]$ consists of $n$ individuals all with welfare $l_2$, there is some population $\Pop[X]$ in which everyone has welfare $l_1$ and some population $\Pop[Z]$ such that $\Pop[X + Z]$ is better than $\Pop[Y + Z]$. As per the results in \S \ref{section-AveragistViews}, AU/VV1/VV2 support Repugnant Addition with respect to a large enough background population  $\Pop[Z]$ with $\Avg[Z] \leq 0$ (and indeed, when $\Avg[Z] < 0$, they support the much more repugnant conclusion that, for any population $\Pop[Y]$ in which everyone has positive welfare, there is a larger population $\Pop[X]$ in which everyone has negative welfare such that $\Pop[X + Z]$ is better than $\Pop[Y + Z]$).

The difficulty of avoiding Repugnant Addition has been noticed independently by \citeauthor{budolfsonMSwhy} (ms), who provide a thorough exploration of the phenomenon covering a broader range of axiologies than we have considered here. So rather than saying any more about this implication, we direct the reader to their results. 

%todo ADD IN? When the background population has positive welfare, then these theories can avoid the Repugnant Conclusion, but at the cost of the Sadistic Conclusion [cites]. When the background population has negative welfare, of course, they imply an even-more-sadistic conclusion. [Is there a name of this?]

%todo ADD IN? Like PR, the non-separable egalitarian theories we've discussed are even more susceptible to the Repugnant Conclusion than TU.

%\footnote{The fact that many axiologies endorse Repugnant Addition in the presence of large background populations has been noticed independently by \citeauthor{budolfsonMSwhy} (ms). Their paper provides a detailed exploration of the Repugnant Conclusion in the presence of background populations, including discussion of intransitive, incomplete, and person-affecting axiologies axiologies, and of decision-making under risk, topics which we have set aside in this paper.}

\subsection{Infinite ethics}
\label{section-InfiniteEthics}

A long-standing and unresolved challenge for axiology is how to extend axiologies from finite to infinite contexts.\footnote{For surveys of the difficulties of infinite axiology, see for instance \cite{asheim2010intergenerational}, \cite{bostrom2011infinite}, and Ch.\@ 1 of \cite{askell2018pareto}.} Most of the extant proposals for ranking infinite worlds, in both philosophy and economics, aim to extend total utilitarianism.\footnote{See, for instance, \cite{atsumi1965neoclassical}, \cite{diamond1965evaluation}, \cite{von1965existence}, \cite{vallentyne1997infinite}, \cite{lauwers2004infinite}, \cite{bostrom2011infinite}, \cite{arntzenius2014utilitarianism}, \cite{jonsson2018limit}, \citeauthor{wilkinsonFCinfinite} (forthcoming), and \citeauthor{clarkMSinfinite} (ms), among many others.} However, these proposals can easily be adapted to extend other additive axiologies. For instance, a simple extension of total utilitarianism (suggested in \cite{lauwers2004infinite}) simply compares any two populations by summing the differences in welfare between the two populations for each individual, treating an individual who doesn't exist in a population as having welfare 0.\footnote{Formally, $\Pop[X] \succcurlyeq \Pop[Y]$ if and only if $\sum{p_i \in X \cup Y} w_\mathbf{x}(p_i) - w_\mathbf{y}(p_i)$ converges unconditionally to a value $\geq 0$, where for any $p_i \not \in X$, $w_\mathbf{x}(p_i) = 0$ (and likewise for $\Pop[Y]$).} This axiological criterion can easily be adapted to a critical-level or prioritarian theory by simply replacing welfare with some increasing function of welfare.

It is much less clear, however, how to extend non-additive theories to the infinite context, and there has so far been little if any discussion of this question. Our limit results, however, suggest a partial answer: when comparing two infinite populations, at least when these populations differ only finitely, we are quite literally in (and not merely approaching) the large-background-population limit. So it is natural to think that a non-additive axiology $\mathcal{A}$ that has an additive counterpart $\mathcal{A}'$ %with which it asymptotically agrees as the size of the background population goes to infinity 
should agree exactly with that additive counterpart in the infinite context. For instance, if we are average utilitarians and we live in an infinite world, but we can only affect a finite part of that infinite world, then we should simply compare the possible outcomes of our choices by the appropriate infinite generalization of critical-level utilitarianism, where the critical level is the average welfare level in the background population.

%[TT] TODO feels like there's some step in the exposition missing - not clear why we jumpt to talking about relative frequencies.
This suggestion is well-defined %todo or "intelligible"
only if we have a well-defined notion of \textit{relative frequency} for infinite worlds---specifically, the relative frequency of different welfare levels in an infinite population, which lets us make sense of further notions like a \textit{welfare distribution} and \textit{average welfare}. A natural suggestion here (advocated, for instance, by \cite{knobe2006philosophical}) is to use the \textit{limiting} relative frequency in uniformly expanding spatiotemporal regions, providing that this limit exists and is the same for all starting locations. There is plenty of debate to be had about this proposal, but this is not the place for that debate. At any rate, it seems plausible (though far from indisputable) that there should be \textit{some} way of making sense of the relative frequencies of particular welfare levels in an infinite population. %todo OR: ...describing the welfare distribution of an infinite population.

\subsection{Opportunities for manipulation}

The results in \S\S \ref{section-AveragistViews}--\ref{section-EgalitarianViews} have one other interesting implication: they suggest a way in which agents who accept non-separable axiologies can be \textit{manipulated}. Suppose, for instance, that we in the Milky Way are all average utilitarians, while the inhabitants of the Andromeda Galaxy are all total utilitarians. And suppose that, the distance between the galaxies being what it is, we can communicate with each other but cannot otherwise interact. Being total utilitarians, the Andromedans would prefer that we act in ways that maximize total welfare in the Milky Way. To bring that about, they might create an enormous number of welfare subjects with welfare very close to zero---for instance, breeding quintillions of very small, short-lived animals with mostly bland experiences---and send us proof that they have done so. We in the Milky Way would then make all our choices under the awareness of a large background population whose average welfare is close to zero. If they could create for us a large enough background population with average welfare sufficiently close to zero, the Andromedans could move us arbitrarily close to \textit{de facto} total utilitarianism.

It's not obvious whether such a strategy would be efficient, but it might be, if creating small, short-lived welfare subjects with bland experiences (and transmitting the necessary proof of their existence) is sufficiently cheap. Since the cost of creating a welfare subject with welfare $x$ presumably increases with $|x|$ (and plausibly increases at a super-linear rate), it might well make sense for the Andromedans to devote some of their resources to this manipulation strategy rather than spending all their resources directly on creating welfare subjects with high welfare.

%todo ADD SOME OF THIS RHETORIC BACK IN? The limit result provides a way for totalists to manipulate averagists: Just create a huge number of welfare subjects with $\sim 0$ welfare, tell the averagists you're doing it, and eventually they'll have to start acting like totalists.

As the preceding results demonstrate, this kind of manipulability is not unique to average utilitarians, but applies also to agents who accept variable-value or non-separable egalitarian views.\footnote{But manipulating egalitarians may be more expensive, if it requires creating beings with a wide distribution of welfare levels. Likewise, agents who accept a critical-level view other than TU may find it more expensive to manipulate in this way, since they may need to create welfare subjects at or near what they regard as the critical level---unless, for instance, creating welfare subjects with welfare close to zero can reduce the average welfare of a pre-existing background population toward that critical level.} Moreover, the potential for manipulation is not symmetric: since the Andromedans accept a separable axiology, what they choose to do in their galaxy will not be affected by their beliefs about what we are doing in ours (except in the ordinary ways, involving potential causal interactions between our galaxies). %todo EXPLAIN THIS BETTER!

Diverting though these speculations might be, the real-world opportunities for this sort of axiological manipulation may be quite limited. Setting aside the likelihood of nearby planets or galaxies being monopolized by proponents of rival axiologies, if there is a large enough pre-existing background population in the universe as a whole (say, outside the region accessible either to us or to the Andromedans), then it may be very hard for the Andromedans to have any significant impact on the size or welfare distribution of the background population. This might be welcome news for them: if the average welfare of the background population is already close to zero, then they will get what they want from us averagists, without having to work for it. But if the average welfare in the background population is non-zero, then we may not behave quite as the Andromedans would most prefer.

This illustrates a general point: the preceding arguments are not necessarily good news for total utilitarians, or for proponents of any other separable axiology in particular. In the presence of large background populations, non-separable axiologies can converge with a wide range of separable counterparts, which disagree among themselves about how to rank populations and how to act for the best. So although large background populations generate \textit{some} convergence among axiologies on particular practical conclusions, axiological disputes remain practically significant.

\section{Conclusion}
\label{section-Conclusion}

We have shown that, in the presence of large enough background populations, a range of non-additive axiologies asymptotically agree with some counterpart additive axiology (either critical-level or, more broadly, prioritarian). And we have argued that the real-world background population is large enough to make these limit results practically relevant. The most notable implication of these arguments is that `arguments from astronomical scale'---in particular, for the overwhelming importance of %near-future
existential catastrophes---need not depend on an assumption of axiological separability.

We have left many questions unanswered that might be valuable topics of future research: (1) a more careful characterization of the size and welfare distribution of the real-world background population; (2) the significance of risk/uncertainty, particularly with respect to these characteristics of the background population; (3) the behavior of a wider range of non-additive axiologies (e.g.\@ incomplete, intransitive, or person-affecting) in the large-background-population limit; and (4) exploring more generally the question of how large the background population needs to be for the limit results to `kick in', for a wider range of axiologies and choice situations than we considered in \S \ref{section-ExistentialCatastrophe}.

%under what general conditions the large-background-population limit results kick in for particular non-additive axiologies, e.g., whether it is generally sufficient that the background population be significantly larger than one or both of the populations under comparison.

%todo ADD SOMETHING LIKE THIS BACK IN Uncertainty provides a moderate counterweight to these conclusions: \textit{Expectational} average utilitarianism attaches more decision weight to states with a smaller background population. (Probably, expectational average utilitarianism implies \textit{de facto} egoism, since the differences in expected average welfare between our available acts are mainly driven by the tiny probability that solipsism is true.)

%todo ACKNOWLEDGEMENTS FOOTNOTE
%\ifanon
%\else
%\footnote{For helpful discussion and/or feedback on earlier versions of this paper, we are grateful to NAMES.}
%\fi

%todo ACKNOLWEDGEMENTS DETAILS [
%GPI `brainstorm session' during Trinity Term 2019
%%Tomi Francis and Itzhak Rasooly in particular provided useful feedback before/during/after this meeting.
%Kacper - Extended conversation in November 2019.
%Toby - He took a look at a draft, and we had a one-hour meeting about it, in November 2019.
%GPI End-of-Term workshop, Michaelmas Term 2019
%Phil Trammell - sent some brief written comments by email c. November 2019
%Hilary Greaves - Significant written comments and back-and-forth by email, 2020 9
%Tomi Francis - Extensive written comments, 2020 10
%todo ]

\appendix

\section{Results}

Recall that $\WW$ is the set of welfare levels, and $\PP$ consists of all non-zero, finitely supported functions $\WW\to\ZZZ_+$. By a \emph{type} of populations we mean a set $T\subset \WW$ that contains populations of 
%[TT] Not sure we strictly need this. TODO: think about it, maybe
arbitrarily large size:  for all $n \in \NNN$ there exists $X\in T$ with $\Size[X]\geq n$.

The following result, while elementary, indicates our general method.
\begin{lemma}\label{l:gen}
	 Suppose given $V\colon\PP\to \RRR$ and a positive function $s\colon \NNN\to\RRR$. %[CT Sept 30] Is there a typo/missing word here? Should it be "Suppose *we are* given..." or something like that?
	 %[CT Sept 30] Also, why not "$s\colon \NNN\to\RRR_{++}$" (I take it that's what "positive function" means here?) [TT] IDK, seems fine to me :)
	 Define 
	\[
		V\marginal(X)\coloneqq \lim_{\Size[Z]\to\infty} \bigl(V(X+Z)-V(Z)\bigr)s(\Size[Z])
	\]
as $Z$ ranges over populations of some type $T$. 
%
%[TT] Don't really need this in the thingy
%Then, for all sufficiently large populations of type $T$,
%\[ V\marginal(X+Z)>V\marginal(Y+Z)\implies  V(X+Z)>V(Y+Z).
%\]
If the axiology with value function $V\marginal$ is separable, then the axiology with value function $V$ converges to it, relative to background populations of type $T$.
\end{lemma} 
\noindent

\begin{proof}
Let $Z$ be a background population of type $T$.
Suppose that $V\marginal(X+Z)>V\marginal(Y+Z)$. Given that the corresponding axiology is separable, we must have $V\marginal(X)>V\marginal(Y)$. 
Then, if $\Size[Z]$ is large enough,  
	 \[
		 \bigl(V(X+Z) - V(Z)\bigr){s(\Size[Z])} > \bigl(V(Y+Z) - V(Z)\bigr){s(\Size[Z])},
	 \]
	 whence, rearranging, $V(X+Z) > V(Y+Z)$. %[CT Sept 30] Just noting that Lemma 1 and its significance are much clearer to me now than they were previously!
\end{proof} 

\AveragismThm*

\begin{proof}
	 In this case, a brief calculation shows
	\begin{equation}\label{eq:avgdif}
		 V\sAU(X+Z)-V\sAU(Z)
		 =\frac{(\Avg[X]-\Avg[Z])\Size[X]}{\Size[X]+\Size[Z]}
		 =\frac{V\sCL[c](X)}{\Size[X]+\Size[Z]}.
	\end{equation}
	Setting $s(n)=n$ we find 
		$V\sAU\marginal(X)=V\sCL[c](X)$, 
	in the notation of Lemma \ref{l:gen}. That Lemma then yields the first statement. 
	
	We now verify the more precise second statement directly. Suppose $\Avg[Z]=c$, that \eqref{eq:avg1} holds, and that $V\sCL[c](X)>V\sCL[c](Y)$.  
    We have to show $V\sAU(X+Z)>V\sAU(Y+Z)$. 
    Using \eqref{eq:avgdif}, that desired conclusion is equivalent to 
\[
    \frac{V\sCL[c](X)}{\Size[X]+\Size[Z]}>
    \frac{V\sCL[c](Y)}{\Size[Y]+\Size[Z]}.
\]
Cross-multiplying, this is equivalent to
\[
    {V\sCL[c](X)}(\Size[Y]+\Size[Z])>
    {V\sCL[c](Y)}(\Size[X]+\Size[Z])
\]
or, rearranging,
%[TT] Argument based on general result:
%The value function $\VAv(X)=\Avg[X]$ of average utilitarianism is scale monotonic, with
%\[
%\nabla_X\VAv(Z)
%=\lim_{t\to 0} \frac{\frac{\Avg[Z]\Size[Z]
%	+\Avg[X]t\Size[X]}{\Size[Z]+t\Size[X]}-\Avg[Z]}{t}
%=\lim_{t\to 0} \frac{(\Avg[X]-\Avg[Z])\Size[X]}{\Size[Z]
%	+t\Size[X]}
%=\frac{(\Avg[X]-\Avg[Z])\Size[X]}{\Size[Z]}
%=\frac{\VCL[{\Avg[Z]}](X)}{\Avg[Z]}.
%\]
%Thus, for all $n$ big enough,
%\[
%	\VCL[{\Avg[Z]}](X) > \VCL[{\Avg[Z]}](Y) 
%	\Rightarrow \nabla_X\VAv(Z) > \nabla_X\VAv(Y) 
%	\Rightarrow \VAv(nZ+X)\sGeq \VAv(nZ+Y).
%\]
	 \begin{equation}\label{eq:avg3}
		\Size[Z](V\sCL[{c}](X)-V\sCL[{c}](Y)) 
		> \Size[X]V\sCL[{c}](Y)-\Size[Y]V\sCL[{c}](X). 
	\end{equation}
	Given that   $V\sCL[{c}](X)-V\sCL[{c}](Y)>0$,
	the desired conclusion \eqref{eq:avg3} follows from \eqref{eq:avg1}. %[CT Sept 30] Switched the order here (was previously "\eqref{eq:avg3} follows from \eqref{eq:avg1}"), as I think was intended. [TT] Ah, no it's the other way around! I agree the logic of the proof is not ideally laid out. Switching back, and adding "the desired conclusion" to, I don't know, clarify [CT Sept 30] Oh, my bad -- I see what's going on now!
\end{proof} 

\VVThm*

\begin{proof}
	Suppose the variable value view has a value function of the form 
	$V(X)=f(\Avg[X])g(\Size[X])$. Then
	\begin{equation*}
	\begin{aligned} 
		V(X+Z)-V(Z) = {}&{} f(\Avg[X+Z])g(\Size[X]+\Size[Z]) -f(\Avg[Z])g(\Size[Z])\\
		= {}&{} f(\Avg[X+Z])\bigl(g(\Size[X]+\Size[Z])-g(\Size[Z])\bigr)\\
		&\qquad
		  + \bigl( f(\Avg[X+Z])-f(\Avg[Z]) \bigr) g(\Size[Z]).
	\end{aligned}
	\end{equation*}
We now apply two lemmas, proved below.
	\begin{lemma}\label{L:g}
		We have $\bigl(g(\Size[X+Z])-g(\Size[Z])\bigr)\Size[Z]\to 0$ as $\Size[Z]\to\infty$.
	\end{lemma}
	\begin{lemma}\label{L:f}
		We have $\bigl(f(\Avg[X+Z])-f(\Avg[Z])\bigr)\Size[Z]\to 
	f'(c)V\sCL[c](X)$ 
		 as $\Size[Z]\to\infty$ with $\Avg[Z]=c$.
	\end{lemma}
	\noindent
	Since $f(\Avg[X+Z])\to f(c)$, and $g(\Size[Z])$ approaches some upper bound $L$ as 
	$\Size[Z]\to \infty$, we find
	\[
		\lim_{\Size[Z]\to\infty}\bigl(V(X+Z)-V(Z)\bigr)\Size[Z] = 
		 f'(c)V\sCL[{c}](X) L
	\]
	as $Z$ ranges over populations with $\Avg[Z]=c$. 
	Let $s(n)=\frac{n}{f'(c)L}$. Then we have found
	\[
		\lim_{\Size[Z]\to\infty}\bigl(V(X+Z)-V(Z)\bigr)s(\Size[Z]) =  V\sCL[{c}](X).
	\]
    The result now follows from Lemma \ref{l:gen}. 
\end{proof}

\begin{proof}[Proof of Lemma \ref{L:g}] %[CT Sept 30] I found the order a little confusing here, since I expected the proofs of the lemmas to come immediately after the lemmas were stated (and hadn't yet looked at the next page), so I wondered "Are these just supposed to be self-evident?" Happy to defer to you on this, though. [TT] For now at least, just adding some sign-posting in the above proof
	Let $z$ be the result of rounding $\Size[Z]/2$ up to the nearest integer. By increasingness and concavity of $g$, we have\footnote{The general fact being used about concavity is that, if $x>y>z$, then
	$\frac{g(x)-g(y)}{x-y}\leq 
	\frac{g(y)-g(z)}{y-z}$.
	} 
	\[
		0\leq \frac{g(\Size[X+Z])-g(\Size[Z])}{\Size[X]}
		\leq \frac{g(\Size[Z])-g(z)}{\Size[Z]-z}
		\leq \frac{g(\Size[Z])-g(z)}{\Size[Z]/2}.
	\]
	Cross-multiplying,
	\[
	0\leq {\bigl(g(\Size[X+Z])-g(\Size[Z])\bigr)}{\Size[Z]}
		\leq {2\bigl(g(\Size[Z])-g(z)\bigr)}{\Size[X]}.
	\]
	Since $g(\Size[Z])$ and $g(z)$ both tend to a common limit $L$ as $\Size[Z]\to\infty$, we find that the right-hand side tends to $0$ in that limit. Therefore the expression in the middle also tends to $0$.
\end{proof}

\begin{proof}[Proof of Lemma \ref{L:f}]
	First, if $\Avg[X]=c$ then $f(\Avg[X+Z])-f(\Avg[Z])=0$ and $V\sCL[c](X)=0$, so the result is trivial in that case. Otherwise, since $\Avg[X+Z]$ tends toward $c$ as $\Size[Z]\to\infty$, we have (by the definition of the derivative)
	\[
		\frac{f(\Avg[X+Z])-f(\Avg[Z])}{\Avg[X+Z]-\Avg[Z]}\to f'(c).
	\]
	We have, from \eqref{eq:avgdif},
	\[
		\Avg[X+Z]-\Avg[Z]=\frac{V\sCL[c](X)}{\Size[X]+\Size[Z]}.
	\]
	Inserting this into the preceding formula, we find 
	\[
		(f(\Avg[X+Z])-f(\Avg[Z]))(\Size[X]+\Size[Z])\to f'(c) V\sCL[c](X). 	
	\]
	Since $(f(\Avg[X+Z])-f(\Avg[Z]))\Size[X]\to 0$, we obtain the desired result.  	
\end{proof}

%%%%%%%%%%%%%%%%%%%%%%%%%%%%%%%%%
%%%%%%%%%%%%%%%%%%%%%%%%%%%%%%%%%
%%%%%%%%%%%%%%%%%%%%%%%%%%%%%%%%%
\begin{prop}\label{p:VV1mix}
For any populations $X$ and $Y$, if 
$X\succ\sTU Y$ and
$X\succ\sAU Y$, then $X\succ\sVV{1}Y$.
%If, in addition, $\Avg[X]\geq 0$ and $\Size[X]\geq \Size[Y]$, then $X\succ\sVV{2}Y$.
\end{prop}
\begin{proof}
First, note that $V\sVV{1}(X)$ has the same sign as $\Avg[X]$. So if $\Avg[X]\geq 0 \geq\Avg[Y]$, then it is automatic that $V\sVV{1}(X)>V\sVV{1}(Y)$. (The condition that $X\succ\sTU Y$ and $X\succ\sAU Y$ excludes the case where $\Avg[X] = 0 = \Avg[Y]$.) Thus it remains to consider the case when $\Avg[X]$ and $\Avg[Y]$ are both positive or both negative.

First suppose they are positive. If $\Size[X]\geq \Size[Y]$, then, since $g$ is increasing and $\Avg[X]>\Avg[Y]$, 
%[TT]^ weakly, if necessary
$V\sVV{1}(X)=\Avg[X]g(\Size[X])>\Avg[Y]g(\Size[Y])=V\sVV{1}(Y)$, as required.
If, instead, $\Size[Y]>\Size[X]$, then we have
\[
\frac{V\sVV{1}(X)}{V\sVV{1}(Y)}
=
\frac{\Avg[X]g(\Size[X])}{\Avg[Y]g(\Size[Y])}
\geq
\frac{\Avg[X]\Size[X]}{\Avg[Y]\Size[Y]}
>1
\]
and therefore $V\sVV{1}(X)>V\sVV{1}(Y)$. Here, the first inequality uses the concavity of $g$, and the second the fact that $\Tot(X)>\Tot(Y)>0$.

The case where $\Avg[X]$ and $\Avg[Y]$ are negative is similar, with careful attention to signs.
%
%[TT] Old proof with small gaps, can delete if above survives
%We can extend $g$ to a non-negative, increasing, bounded, concave function $\RRR_+\to\RRR_+$. 
%First suppose $\Avg[X]>\Avg[Y]\geq 0$.
%\[
%\frac{V\sVV{1}(Y)}{\Avg[X]}=\frac{\Avg[Y]}{\Avg[X]}g(\Size[Y])
%<
%g(\frac{\Avg[Y]}{\Avg[X]}\Size[Y])
%\leq
%g(\Size[X])=\frac{V\sVV{1}(X)}{\Avg[X]}
%\]
%Here the two equalities uses the definition of $V\sVV{1}$; the first inequality uses the 
%%[TT]??? concavity=strict? TODO - sort this out
%concavity of $g$; the second inequality uses the fact that $\Avg[Y]\Size[Y]=\Tot(Y)<\Tot(X)=\Avg[X]\Size[X]$, and the fact that $g$ is increasing.
%It follows that $V\sVV{1}(Y)< V\sVV{1}(X)$. 
%
%If, instead, $0\geq \Avg[X]>\Avg[Y]$, then, by similar reasoning,
%\[
%\frac{V\sVV{1}(X)}{\Avg[Y]}=\frac{\Avg[X]}{\Avg[Y]}g(\Size[X])
%<
%g(\frac{\Avg[X]}{\Avg[Y]}\Size[X])
%\leq
%g(\Size[Y])=\frac{V\sVV{1}(Y)}{\Avg[Y]}
%\]
%whence again
%$V\sVV{1}(X)>V\sVV{1}(Y)$.
%
%Finally,  if $\Avg[X]>0>\Avg[Y]$ then it is automatic that $V\sVV{1}>0>V\sVV{1}(Y)$. 
%%
%
%
%
%\begin{align*}
%V\sVV{1}(X)- V\sVV{1}(Y) &=
%\Avg[X]g(\Size[X])-
%\Avg[Y]g(\Size[Y])\\
%&=
%\Avg[X]\bigl(g(\Size[X])-g(\Size[Y])\bigr)+
%\bigl(\Avg[X]-\Avg[Y])\bigr) g(\Size[Y])
%\end{align*}
%
\end{proof}

\GenEgalThm*

%[TT] TODO Could do with some work making this more reasonable
\begin{rem}\label{rem:derivative}
Before proving Theorem \ref{t:GenEgal}, we should explain the requirement that `$I$ is a differentiable function of the distribution of $X$'. It has two parts.  First, let $\PP_\RRR$ be the set of finitely-supported, non-zero functions $\WW\to\RRR_+$. Let $\DD\subset\PP_\RRR$ be the subset of distributions, i.e.~those functions that sum to $1$. The first part of the requirement is that there is a function $\tilde I\colon \DD\to\RRR$ such that $I(X)=\tilde I(X/\Size[X])$. In that sense, $I(X)$ is just a function of the distribution of $X$. Another way to put this is that $I$ can be extended to a function on all of $\PP_\RRR$ that is scale-invariant, i.e. $I(nX)=I(X)$ for all reals $n>0$ and all $X\in\PP_\RRR$. The second part of the requirement is that $I$, so extended, is differentiable, in the following sense:%   
\footnote{This can also be interpreted as a differentiability requirement directly on $\tilde I$: it should have a linear G\^ateaux derivative.  }
for all $P,Q\in\PP_\RRR$, the limit
\[
\partial_Q I(P)\coloneqq \lim_{t\to 0^+} \frac{I(P+tQ)-I(P)}{t} \]
exists and is linear as a function of $Q$. 
%$\partial_{Q+R}I(P)=\partial_{Q}I(P)+\partial_{R}I(P)$. 
In effect, $Q\mapsto \partial_Q I(P)$ is the best linear approximation of $I-I(P)$. In practice we only need $I$ to be differentiable at the background distribution $D$.
%(We note that, instead of additivity, the stronger condition of linearity is usually required for differentiability.)
\end{rem}
%[TT] TODO: Maybe we should just have I(Z/|Z|) in the value function - explicitly depending only on Z/|Z|. Instead of this weird explanation above, extend I from (probability) distributions to \PP_\RRR, and assume the result is differentiable in the stated sense. Then in the expressions for f, and in the theorem, use D as the fixed background distribution.

\begin{proof} 
	Let $Z$ range over background populations with the given distribution $D=Z/\Size[Z]$. Thus $Z$ is of the form 
    $n D$ for some  $n>0\in\RRR$.   

	Define $s(n)=1$, in the case of TU-based egalitarianism, and $s(n)=n$ in the case of AU-based egalitarianism. 
Noting that value functions of the assumed form can be evaluated not only on $\PP$ but on the larger set $\PP_\RRR$ (see Remark \ref{rem:derivative}), we have  
\[
V(nX)=(n/s(n))V(X) .
\]

We can then see that $V\marginal$ (as defined in Lemma \ref{l:gen}) is the directional derivative of $V$ at $D$: 
\[
	\begin{aligned}
	V\marginal(X)
	&{} = \lim_{\Size[Z]\to\infty} \bigl(V(Z+X)-V(Z)\bigr) s(\Size[Z])\\
	&{} = \lim_{n\to\infty} \bigl(V(nD+X)-V(nD)\bigr) s(n)\\
	&{} = \lim_{n\to\infty} \frac{V(D+\tfrac 1{n}X)-V(D)}{1/n} \eqqcolon \partial_{X} V(D). %[CT Sept 30] Not obvious to me why s(n) goes away on the final line (specifically, in the case of AU-based egalitarianism -- shouldn't it cancel out 1/n in the denominator? [TT] I think it's right (roughly, it *becomes* the 1/n)
	\end{aligned} 	
\]
For totalist egalitarianism, we find that 
\[
	V\marginal(X) = \Tot(X) - \partial_X I(D) - I(D)\Size[X].
\]
Given that $I$ is differentiable as in Remark \ref{rem:derivative}, 
this function is additive in $X$ and therefore represents an additive axiology $\Ax'$. More specifically, for each welfare level $w$ let $1_w$ be a population with one person at level $w$. We then have
\[
	V\marginal(X) = \sum_{w\in\WW} X(w)f(w)
	\quad\mbox{with}\quad
	f(w)=w  - \partial_{1_w} I(D)-I(D).
\]
Similarly, for averagist egalitarianism,
%[TT] TODO Worth checking again: (not that we need these details, actually)
\[
	\begin{aligned}
		V\marginal(X) &= (\Avg[X]-\Avg[D])\Size[X] - \partial_X I(D)\\
					  &= \sum_{w\in\WW} X(w)f(w) \quad\mbox{with}\quad f(w)=w -\partial_{1_w}I(D) -\Avg[D].
	\end{aligned}
\]
Now, suppose $X^+$ differs from $X$ in that one person is better off, say with welfare $v$ instead of $w$. If the Pareto principle holds with respect to $\Ax$, then $V(X^++Z)\geq V(X+Z)$ for all $Z$; by convergence, we cannot have $V\marginal(X^+)<V\marginal(X)$. It follows that $f(v)\geq f(w)$; thus $f$ is weakly increasing. By the same logic, Pigou-Dalton transfers do not make things worse with respect to $\Ax'$, and it follows that $f$ is weakly concave. 
%[TT] ANNOYING STUFF - related to comments before statement of t:GenEgal. Fix?
%[TT] But if we define prioritarianism with just "weakly concave" and "weakly increasing" then it includes the trivial ax, and everything converges to that.
%[TT] Might be able to say "weakly but non-trivially..." - but that's a bit me, might be trivial practically everywhere
%[TT] Maybe the best thing to do is to state the formulae, and say *this* is Prioritarian?
%[TT] Possibility: In def of PR, make f non-constant (but could be locally constant); in def of Egal, make I not monotonically increasing  - ie. make it possible to decrease inequality by benefitting the worst off. 
\end{proof}

\MDTThm*

\begin{proof} 
	Define $\Pair{X,Y}=\sum_{x,y\in \WW} X(x)Y(y)|x-y|$. Then $\MAD(Z)=\Pair{Z,Z}/\Size[Z]^2$. 
	It is easy to check that 
	$\partial_X\Pair{Z,Z} = 2\Pair{X,Z}$
	and therefore
	\[
		\partial_X\MAD(Z) = 2\frac{\Pair{X,Z}}{\Size[Z]^2} - 2\frac{\Pair{Z,Z}}{\Size[Z]^3}\Size[X].
	\]
	In particular, $\MAD$ is differentiable and Theorem \ref{t:GenEgal} applies.
%[TT] TODO? Only really need differentiability at D - that would allow us to cut out some equations
	Following the proof of Theorem \ref{t:GenEgal}, we know that MDT converges to the additive axiology $\Ax'$ with weighting function
	\begin{align*}
			f(w) &= w
			- \alpha \partial_{1_w} \MAD(D)
			-\alpha \MAD(D)  
			\\
			 &= w
			 - 2\alpha \Pair{1_w,D}
			 -\alpha \MAD(D)  
			 \\
				 %&= w-\alpha \MAD(Z)  - 2\alpha \Pair{1_w,Z}/\Size[Z] + 2\alpha\Pair{Z,Z}/\Size[Z]^2\\
				% &= w- 2\alpha \Pair{1_w,Z}/\Size[Z] + \alpha \MAD(Z)\\
				 &= w- 2\alpha \MAD(w,D) + \alpha \MAD(D)
	%			 &= w- 2\alpha \sum_{x<w} \frac{Z(x)}{\Size[Z]}(w-x) - 2\alpha \sum_{x>w} \frac{Z(x)}{\Size[Z]}(x-w) + \alpha \MAD(Z)
.  \qedhere \end{align*}
\end{proof} 

\QAAThm*
\begin{proof}
	Theorem \ref{t:GenEgal} applies, with $I(X)=\Avg[X]-\QAM(X)$. (We omit the proof that this $I$ is differentiable.) We have, then, convergence to prioritarianism with a priority weighting function
	\[
		\begin{aligned}
			f(w) &= \partial_{1_w}\QAM(D) =  
			\frac{g(w)-\sum_{x\in\WW} D(x)g(x)}
			{g'(\QAM(D))} .
		\end{aligned}
%[TT] TODO Could probably add more detail
	\]
	Since the background distribution $D$ is fixed, this differs from the stated priority weighting function only by a positive scalar (i.e. the denominator).
\end{proof} 

\positiveRDThm*

\begin{proof}

	Suppose that the weighting function $f$ has a horizontal asymptote at $L>0$. As in Lemma \ref{l:gen} it suffices to show that $\lim_{\Size[Z]\to\infty} V(X+Z)-V(Z)= L\Tot(X)$, as $Z$ ranges over populations with distribution $D$, and on the assumption that $X$ is moderate with respect to $D$. 
	
	 Write $X_{\leq w}=\sum_{x\leq w} X(w)$ for the number of people in $X$ with welfare at most $w$, and similarly $X_{<w}=\sum_{x<w}X(w)$. Separating out contributions from $X$ and contributions from $Z$, we have 
	\[
		\begin{aligned}
			V(X+Z)-V(Z) &= %\sum_{w\in\WW} \biggl(\sum_{i=1}^{X(w)+Z(w)} f(Z_{<w}+X_{<w} +i) - \sum_{i=1}^{Z(w)} f(Z_{< w}+i)  \biggr)w\\
			%&= 
			\sum_{w\in\WW}
			\sum_{i=1}^{X(w)} f(Z_{\leq w}+X_{<w}+i)w
			\\&\qquad+
			\sum_{w\in\WW}\sum_{i=1}^{Z(w)} \bigl(f(Z_{<w}+X_{<w} +i) -  f(Z_{< w}+i)  \bigr)w.
		\end{aligned} 
	\]
The assumption that $X$ is moderate means that,
in those cases where $X(w)\geq 1$, so that the first inner sum is non-trivial,  we also have $Z_{\leq w}\to\infty$. We see therefore that each summand in the first double-sum tends to $Lw$. The first double sum then converges to $\sum_{w\in\WW} X(w)Lw=L\Tot(X)$. It remains to show that the second double sum converges to $0$. Call the summand in that double sum $S(w,i)$. 

Since there are finitely many $w$ for which $Z(w)\geq  1$ (making the inner sum non-trivial), it suffices to show that, for each such $w$, the inner sum converges to $0$. If $X_{<w}=0$, then the inner sum is identically zero, so we can assume $X_{<w}\geq 1$.  
We can also assume that $Z_{<w}$ is large enough that $f$ is convex in the relevant range; then 
\[
0\leq S(w,i)\leq \bigl(f(Z_{<w}+X_{<w})-f(Z_{<w})\bigr)w.
\]
Moreover, the number of terms, $Z(w)$, is proportional to $Z_{<w}$. It remains to apply the following elementary lemma with $n=Z_{<w}$ and $m=X_{<w}$.
\begin{lemma}\label{l:BRD}
If $f$ is an eventually  convex function decreasing to a finite limit,  then $n(f(n+m)-f(n))\to 0$ as $n\to\infty$. 
\end{lemma} 
\noindent This is just a small variation on Lemma \ref{L:g}, and we omit the proof.
%[TT] TODO Combine them cleverly or something; at least be a bit more consistent about presentation 
\end{proof}
\GDThm*

\begin{proof}
Suppose $X$ and $Y$ are supported on $W$, and $X\succ\sCLL Y$. Let $Z$ be a population with distribution $D$, so $Z=nD$ for some $n>0$. We have to show that $X+Z\succ\sGRD Y+Z$ for all $n$ large enough.

Let $\tilde X$ and $\tilde Y$ be populations of equal size, obtained from $X$ and $Y$ by adding people at the critical level $c$. The assumption that $X\succ\sCLL Y$ means that, for the first $m$ such that $\tilde X_m\neq \tilde Y_m$, we have $\tilde X_m>\tilde Y_m$. This shows that $\tilde Y_m <c$, so that in fact $\tilde Y_m=Y_m$. For brevity define $w\coloneqq Y_m$.

Let $v$ be the next welfare level occurring in $X+Y$ above $w$. If there is no such welfare level, then define $v=c+1$.
We can decompose $Z$ (and similarly for other populations) as $Z=Z_-+Z_{w}+Z_0+Z_+$, where $Z_-$ only involves welfare levels in the interval $(-\infty,w)$, $Z_{w}$ involves only $w$, $Z_0$ only involves welfare levels in in $(w,v)$, and $Z_+$ only involves those in $[v,\infty)$.
%Note that $X_-=Y_-$. We have
%\[
%X+Z = X_- + Z_-+X_0+Z_0+X_++Z_+
%\]
%and correspondingly
Note that $X_-=Y_-$ and $X_0=Y_0=0$ but (because $D$ covers $W$ %[CT Sept 30] Added second clause to the parenthetical, because it took me a while to remember this. (I was thinking "couldn't you have a case where w is the highest welfare level in X or Y, and the next highest welfare level in Z is greater than c + 1?") [TT] Seems right, double check later
and is only supported up to $c$) $Z_0\neq 0$. We have
\begin{align*}
    V(X+Z)=V(X_{-}+Z_{-}+Z_{w})
    &+\beta^{\left|X_-+Z_-+Z_{w}\right|}V(X_{w})\\ 
    &+\beta^{\left|X_-+Z_-+Z_{w}+X_{w}\right|} V(Z_0)\\
    &+\beta^{\left|X_-+Z_-+Z_{w}+X_{w}+Z_0\right|} V(X_++Z_+).
\end{align*}
A similar expression holds for $Y$ in place of $X$. Therefore
\begin{align*}
\frac{V(X+Z)-V(Y+Z)}{\beta^{\left|X_-+Z_-+Z_{w}\right|}}&= V(X_{w})-V(Y_{w})+
(\beta^{\Size[X_w]}-\beta^{\Size[Y_w]}) V(Z_0) + R
\end{align*}
where the remainder $R$ is such that $\lim_{n\to\infty} R=0$. 
%[TT] TODO maybe more detail^
Now we use the standard fact that $\sum_{i=1}^m\beta^i = \beta \frac{1-\beta^{m}}{1-\beta}$.
It follows that 
$V(X_{w})-V(Y_{w})=\beta \frac{\beta^{\Size[Y_w]}-\beta^{\Size[X_w]}}{1-\beta} w$. Therefore
\begin{align*}
\frac{V(X+Z)-V(Y+Z)}{\beta^{\left|X_-+Z_-+Z_{w}\right|}}&= 
(\beta^{\Size[X_w]}-\beta^{\Size[Y_w]})(V(Z_0)-\frac{\beta w}{1-\beta})+R.
\end{align*}
Note that $\beta^{\Size[X_w]}-\beta^{\Size[Y_w]}>0$. 
To conclude that $V(X+Z)>V(Y+Z)$ for all $n$ large enough, it suffices to show that  
\[
\lim_{n\to\infty} V(Z_0)>\frac{\beta w}{1-\beta}.
\]
In fact, if $v'$ is the lowest welfare level greater than $w$ occurring in $D$, then $v'\in(w,v)$ and $\lim_{n\to\infty} V({Z_0})=\frac{\beta v'}{1-\beta}$.
\end{proof} %[CT Sept 30] Just noting that I found myself struggling to follow this proof, especially towards the end. Could just be my lack of mental energy (after working through the rest of the proofs), but maybe there's a way to unpack/explain things a bit more? [TT] Agreed, it's a pain the arse

\bibliographystyle{chicago}
\bibliography{NonSeparableAxiologiesCites}
\end{document}